\documentclass[11pt]{article}
\usepackage{amssymb}
\usepackage{amsmath}
\usepackage{amsthm}
\usepackage{multirow}
\usepackage{color}
\usepackage{booktabs}
\usepackage{array}
\usepackage{comment}
\usepackage[toc,page]{appendix}
\usepackage[top=2.5cm,bottom=2.5cm,left=3cm,right=3cm]{geometry}
\usepackage[hidelinks]{hyperref}
\usepackage{setspace} 
\numberwithin{equation}{section}
\usepackage{enumitem}   
\numberwithin{equation}{section}

\newtheorem{theorem}{Theorem}

\newtheorem{corollary}[theorem]{Corollary}
\newtheorem{definition}[theorem]{Definition}
\newtheorem{example}[theorem]{Example}
\newtheorem{lemma}[theorem]{Lemma}

\newtheorem{remark}[theorem]{Remark}

\usepackage{graphicx}
\allowdisplaybreaks[3]

\title{{Geometric aspects of non-homogeneous $1+0$ operators }}
\author{Marta Dell'Atti$^1$, \qquad Alessandra Rizzo$^2$, \qquad Pierandrea Vergallo$^{3,4}$
\\[5mm]
 \small $^1$  Faculty of Mathematics Informatics and Mechanics\\ 
\small University of Warsaw, Banacha 2, Warsaw 02-097, Poland  \\
\small \texttt{m.dell-atti@uw.edu.pl}\\
\small $^2$  Department of Mathematics\\
\small  University of Palermo, via Archirafi 34, Palermo, Italy \\
\small  \texttt{alessandra.rizzo07@unipa.it} \\
\small $^3$  Department of Basic and Applied Sciences\\
\small  University of Basilicata, via dell'Ateneo Lucano, 85100 Potenza, Italy\\
\small $^4$ Istituto Nazionale di Fisica Nucleare, Sezione di Napoli\\
\small Via Cintia, 80126 Napoli, Italy\\
\small  
\texttt{pierandrea.vergallo@unibas.it}
\date{} 
}
\begin{document}
\maketitle

\begin{abstract}
    Led by the key example of the Korteweg-de Vries equation, we study pairs of Hamiltonian operators which are non-homogeneous and are given by the sum of a first-order operator and an ultralocal structure. We present a complete classification of the Casimir functions associated with the degenerate operators in two and three components. 
    We define tensorial criteria to establish the compatibility of two non-homogeneous operators and show a classification of pairs for systems in two components, with some preliminary results for three components as well. 
    Lastly, we study pairs composed of non-degenerate operators only, introducing the definition of bi-pencils. First results show that the considered operators can be related to Nijenhuis geometry, proving a compatibility result in this direction in the framework of Lie algebras. 
\end{abstract} 

\section*{Introduction}
\addcontentsline{toc}{section}{Introduction}
The Hamiltonian formalism for differential equations is a consolidated framework in mathematical physics, differential geometry, and field theories, 
which serves as a fundamental technical tool to investigate nonlinear phenomena \cite{DubKri}. In addition, it plays a fundamental role in the theory of integrable systems of differential equations, both ordinary and partial. As shown by Magri in \cite{Magri:SMInHEq}, finding 
a suitable bi-Hamiltonian structure for a given system represents a possible way to prove its integrability and, by using Lenard-Magri's chains tool (see e.g.\ \cite{dorfman91:_local}), a concrete way to 
produce infinitely many conserved quantities and commuting flows.
The bi-Hamiltonian formalism is established by the existence of two different Hamiltonian structures for the same system. One crucial property for a pair of operators to constitute such a structure is their compatibility: any linear combination of the two Hamiltonian operators is Hamiltonian as well.

Starting from Dubrovin and Novikov's pioneering results on Hamiltonian operators of order $1$ in \cite{DN83}, here we aim to provide a geometric description of a more general Hamiltonian structure, defined as the sum of a Dubrovin-Novikov operator (order $1$) and an ultralocal operator (order $0$). In particular, inspired by some recent results on this topic \cite{casa1,DellAVer1,GubOliSgrVer,Riz1,Ver3}, we study pairs of compatible Hamiltonian operators of type~$(1+0)$. 

The geometric nature of compatible pairs of first-order operators was investigated first by Dubrovin~\cite{dub2}, and continued by Ferapontov \cite{Fer11}, Mokhov~\cite{omokh,Mok14,Mok11,Mok12,Mok13} and in the more recent papers \cite{DMS,FL,LPR}. The problem of classifying  bi-Hamiltonian structures of this type has been further connected to the study of flat pencils of metrics (see the seminal work \cite{Mok13}) but also to the broader theory of Frobenius manifolds, these arising in the context of two-dimensional topological field theories \cite{dub1}. Technically, solving this problem is equivalent to solve a highly nonlinear system which has been proved to be integrable, as shown in~\cite{Fer11,Mok12}.

The complete description of the classification conceals technical difficulties dealing with a high number of constraints. However, thanks to the fact that here we are considering non-homogeneous operators, 
further constraints can be solved compared to the same problem formulated only in presence of {homogeneous}operators (i.e.\ either of order $1$ or $0$).

Here we summarise our main results:
\begin{enumerate}
    \item We present a complete classification of Casimir functions for the Poisson brackets associated with operators of type $(1+0)$ in the degenerate case, in Table \ref{tab:casimir_2} for systems in $2$ components and Table \ref{tab:casimir_3} for of the structures in $3$ components. The non-trivial cases for the operators $\mathcal{C}^{ij}_{3,2}$ and $\mathcal{C}^{ij}_{3,5}$ as introduced in \cite{DellAVer1} are treated separately in section~\ref{sec:casimir_degenerate_case}. The classification completes the description of Casimirs of such operators, extending the recent result obtained in \cite{GubOliSgrVer};
    \item We classify pairs of compatible Hamiltonian operators $(\mathcal{A},\mathcal{B})$ of type~$(1+0)$ in $2$ components with $\mathcal{A}$ non-degenerate in Theorem~\ref{2x2thm}. We obtain new cases of pairs of contravariant metrics, 
    thanks to the nature of the structures under study, as noted in Remark~\ref{rmk:compatibility_non_homog}.
    Moreover, in Remark~\ref{rmk:general_coords} we provide the purely contravariant expression of tensors describing the compatibility of the operators $\mathcal{A}$ and $\mathcal{B}$; 
    \item We present some preliminary results of pairs of operators in $3$ components, being able to fully formulate the discussion related to the order $0$ operator in Lemma~\ref{lemma:ultralocal_3}. This is sufficient to reproduce the bi-Hamiltonian structure in the case of the inverted KdV equation in Example~\ref{ex:kdv_inverted};
    \item We introduce the geometric object of the bi-pencil, i.e.\ pairs of pencils of compatible metrics and compatible Poisson tensors associated with the ultralocal structures. Accordingly, in Theorem~\ref{thm:bi_pencil} we show that they are in one-to-one correspondence with non-homogeneous hydrodynamic operators. A further specification on the introduced structure identifying strong bi-pencils is given in Remark~\ref{rmk:strong_bi_pencil};
    \item We show a preliminary result in terms of the theory of Nijenhuis geometry for non-homogeneous operators. In particular, in Theorem~\ref{nijnonhom} we characterise purely algebraic conditions on the coefficients of $\mathcal{A}$ non-degenerate such that a suitable~$(1,1)$-tensor is Nijenhuis torsionless.
\end{enumerate}

\subsubsection*{Structure of the paper}

In the following, in Section \ref{sect2}, we review some known results on the non-homogeneous operators of interest. In Section \ref{sec3}, we focus on the role of Casimir functions for such operators, presenting a classification for Hamiltonian structures in $n=2,3$ components. In Section \ref{sec4} and \ref{biham_nije}, we show the core results of our paper. In the first we investigate the properties of Hamiltonian pairs, classifying them in $n=2$ components and giving some preliminary results for $n=3$. The computational result is then enriched in Section \ref{biham_nije} by the introduction of the geometric structure of {bi-pencils}, which are in bijection with compatible non-homogeneous pairs. A connection with Nijenhuis geometry is finally discussed, revealing an open question: \emph{is it possible to describe compatibility results on non-homogeneous operators using the modern approach of Nijenhuis geometry?}

\section{Non-homogeneous operators}\label{sect2}
Our object of study is the differential Hamiltonian operator $\mathcal{A}_{(1+0)}$ which is naturally associated with quasilinear evolutionary systems, i.e.\ $n$-component systems of the form
\begin{equation}\label{sys}
    u^i_t=V^i_j(u)\,u^j_x+W^i(u), \qquad i=1,2,\dots n\,,
\end{equation}
where $u=(u^i(x,t))_{1 \le i \le n}$ with $u^i(x,t)$ the field variables of the theory and $(x,t)$ the independent variables. In \eqref{sys} $V^i_j(u)$, transforms as a $(1\,,1)$-tensor, $W^i(u)$ transforms as a covector and both $V^i_j$ and $W^i$ depend on the field variables only, not their derivatives\footnote{ The well-posedeness for the existence of solutions to this type of systems is studied in the classical works~\cite{Rozh1,Rozh2,Liu}}. A natural Hamiltonian operator providing the Hamiltonian description for the system~\eqref{sys} is 
\begin{equation}\label{eq:op_1_0}
    \mathcal{A}_{(1+0)} = \underbrace{\, g^{ij}(u)\,\partial_x+b^{ij}_k(u)\,u^k_x\,}_{\mathcal{A}_{(1)}} + \underbrace{\,\omega^{ij}(u)\, }_{\mathcal{A}_{(0)}} \,,
\end{equation}
with $g^{ij}$, $b^{ij}_k$ and $\omega^{ij}$ depending on the field variables only. The operator $\mathcal{A}_{(1+0)}$ represents the simplest possible extension of the first-order homogeneous operator $\mathcal{A}_{(1)}$, this being a fundamental building block in the theory of Poisson structures. First-order operators $\mathcal{A}_{(1)}$ naturally arise in the study of systems of the form \eqref{sys} setting $W^i = 0$, i.e.\
\begin{equation}\label{eq:hydro_system}
    u_t^i = V^i_j(u)\,u^j_x, \qquad i=1,2,\dots n.
\end{equation}
These systems are known as systems of \emph{hydrodynamic type}, and first-order homogeneous operators are then also referred to as Hamiltonian operators of hydrodynamic type (see e.g.~\cite{tsarev85:_poiss_hamil,tsarev91:_hamil,RizVer1,ManRizVer1}). The homogeneous zero-order operator $\mathcal{A}_{(0)}$ in~\eqref{eq:op_1_0} is represented by the ultralocal structure~$\omega^{ij}$.  In this context, the word \emph{homogeneous} refers to the degree of derivation making use of two natural grading rules:
 $   \text{deg}(\partial_x^k)=k\,,$ and $ \text{deg}(u_{hx})=h\,,$
jointly with the common properties for the degrees. 
Requiring for a differential operator to be homogeneous with respect to derivatives of degree 1, one obtains the standard form 
\begin{equation}\label{dnop_tens}
  \mathcal{A}^{ij}_{(1)} = g^{ij}(u)\,\partial_x+b^{ij}_k(u)\,u^k_x\,, 
\end{equation}
with which is associated the bracket of hydrodynamic type
\begin{equation}\label{dnop_PB}
\begin{split}
    \{\, u^i(x),u^j(y) \,\} &= g^{i j}(u(x)) \, \delta_x (x-y) + b^{i j}_{k} (u(x))\, u^k_x\, \delta (x-y)  \,.
    \end{split}
\end{equation}
In the following, we review the conditions for the operators $\mathcal{A}_{(1)}$,  $\mathcal{A}_{(0)}$ and $\mathcal{A}_{(1+0)}$ to be Hamiltonian. 

\subsection{\texorpdfstring{Hamiltonian property of homogeneous operators $\mathcal{A}_{(1)}$ and $\mathcal{A}_{(0)}$}{a1 a0}}

We briefly recall that a differential operator $\mathcal{A}$ is Hamiltonian if the associated bracket~$\{\,\cdot \,,\, \cdot\,\}_{\mathcal{A}}$ defined as
\begin{equation}
    \{F, G\}_{\mathcal{A}} = \int \frac{\delta f}{\delta u^i}\,\mathcal{A}^{ij} \,\frac{\delta g}{\delta u^j}\, dx \,,
\end{equation}
for any pairs of functionals $F = \int f\, dx$ and $G = \int g\, dx$, is a Poisson bracket, i.e.\ it is skew-symmetric and satisfies the Jacobi identity 
$$\{\{F,G\}_{\mathcal{A}},H\}_{\mathcal{A}}+\{\{G,H\}_{\mathcal{A}},F\}_{\mathcal{A}}+\{\{H,F\}_{\mathcal{A}},G\}_{\mathcal{A}}=0\,,$$
for any additional functional $H= \int h\, dx$. Depending on the form of $\mathcal{A}$, one can specify the Hamiltonian property accordingly. 

Operators of type $\mathcal{A}_{(1)}$ \eqref{dnop_tens} were firstly considered by Dubrovin and Novikov in 
 \cite{DN83}, where the authors focused on the study of the non-degenerate case (i.e.\ $\det{g^{ij}}\neq 0$), revealing the geometric properties of such operators. In particular, they proved that an operator of the form \eqref{dnop_tens} is Hamiltonian 
 if and only if  $g_{ij}=(g^{ij})^{-1}$ is a flat  metric and
 \begin{align}\label{eq:cond_ham_nondeg_first_order}
     &b^{ij}_k=-g^{is}\,\Gamma^j_{sk}\,,
\end{align} 
where $\Gamma^{i}_{jk}$ are Christoffel symbols for $g_{ij}$. Introducing the covariant derivative $\nabla_i$ defined in terms of the Levi-Civita connection of $g_{ij}$, the system \eqref{eq:hydro_system} can be rewritten as 
\begin{equation}\label{eq:non_deg_system_hom_1}
    u^i_t = ( \nabla^i\, \nabla_j\, h(u) ) u^i_x\,,    \qquad i = 1, \dots, n\,,
\end{equation}
where $\nabla^i=g^{ij}\,\nabla_j$, and $h(u)$ is the Hamiltonian density depending on the field variables only.

Releasing the non-degeneracy assumption on the leading coefficient $g^{ij}$ in $\mathcal{A}_{(1)}$ 
Grinberg extended the previous result establishing the following theorem in terms of $g^{ij}(u)$ and $b^{ij}_k(u)$:
\begin{theorem}[\cite{grin}] \label{th:ham_A}
The operator \eqref{dnop_tens} is Hamiltonian if and only if
\begin{subequations}
\begin{align} 
&g^{ij}=g^{ji} \,,\\[1ex]
&\label{eq1thm}\frac{\partial g^{ij}}{\partial u^k} = b^{ij}_k+b^{ji}_k \,, \\[1ex] \label{eq2thm}
&g^{is}b^{jk}_s-g^{js}b^{ik}_s=0 \,, \\[1ex] 
&\label{eq4thm}g^{is}\left(\frac{\partial b^{jr}_{s}}{\partial u^k}-\frac{\partial b^{jr}_{k}}{\partial u^s}\right)+b^{ij}_sb^{sr}_k-b^{ir}_s b^{sj}_k=0 \,, \\[1ex]
\begin{split}\label{eq5thm}
&\displaystyle \sum_{(i, j ,r)} \left(b^{si}_q\left(\frac{\partial b^{jr}_k}{\partial u^s}-\frac{\partial b^{jr}_s}{\partial u^k}\right)+b^{si}_k\left(\frac{\partial b^{jr}_q}{\partial u^s}-\frac{\partial b^{jr}_s}{\partial u^q}\right)\right)=0 \,,
\end{split}
\end{align} \end{subequations}
with the sum over $(i\,,j\,,k)$ is on cyclic permutations of the indices. 
\end{theorem}

We  refer to \cite{sav1} for a classification in 2 and 3 components of degenerate operators $\mathcal{A}_{(1)}$.

The operator $\mathcal{A}_{(0)}$ is represented by a $(2,0)$-tensor\footnote{We recall that zero-order operators are also known as ultralocal operators in the context of the infinite-dimensional systems, whereas traditionally called Poisson tensors for finite-dimensional ones.}
which is Hamiltonian if and only if $\omega^{ij}(u)$ satisfies the following. 
\begin{theorem} \label{th:ham_omega}
The operator $\omega^{ij}(u)$ is Hamiltonian if and only if it forms  a finite-dimensional Poisson structure, i.e.\ it satisfies the conditions
\begin{subequations}
\begin{align}\label{cond1om}
        &\omega^{i j} = - \omega^{j i}, \\
 \label{cond2om}
        &\omega^{i s} \,\frac{\partial \omega^{j k}}{\partial u^s} + \omega^{j s}\, \frac{\partial \omega^{k i}}{\partial u^s} + \omega^{k s} \,\frac{\partial \omega^{i j}}{\partial u^s} =0.
\end{align}
\end{subequations}
\end{theorem}

We stress that it is possible to generalise the notion of homogeneous operators also to higher orders. Indeed, following~\cite{DubrovinNovikov:PBHT}, a homogeneous operator of order $m \ge 1$ reads as  
\begin{align}\label{hhom}
   \begin{split}
\mathcal{A}^{ij}_{(m)}&=a^{ij}\!(u) \,\partial^m_x+b^{ij}_k\!(u)\,u^k_x\,\partial^{m-1}_x+\big(c^{ij}_k\!(u)\,u^k_{xx}+c^{ij}_{k\ell}(u)\,u^k_x\,u^\ell_x\big)\partial^{m-2}_x\\[1ex]
&~~~
+\cdots \,+\big( r^{ij}_k\!(u)\,u^k_{mx} + \dots +r^{ij}_{k_1 \dots \,k_m}\!(u)\,u^{k_1}_x\,\cdots\, u^{k_m}_x\big) \,,
\end{split}
\end{align}  
where the coefficients of the different orders of the differential operator $\partial_x$ depend on the field variables only. The necessary and sufficient conditions for the operators \eqref{hhom} to be Hamiltonian were obtained for $m=2,3$ by Doyle \cite{Doyle} and Potemin \cite{Pote} independently. To the authors' knowledge, for the case $m>3$ an analogous result has not been established yet. For a general description of homogeneous Hamiltonian operators with a differential geometric approach we refer to \cite{mokhov98:_sympl_poiss}.

\subsection{\texorpdfstring{Hamiltonian property of non-homogeneous operators $\mathcal{A}_{(1+0)}$}{a10}}

A larger class of evolutionary systems
is represented in the Hamiltonian formalism by the sum of two or more differential homogeneous operators, for this called \emph{non-homogeneous}.
An operator $\mathcal{A}_{(k+m)}$ is said to be an operator of type $(k+m)$ if there exist two homogeneous operators $\mathcal{A}_{(k)}$ and $\mathcal{A}_{(m)}$ (with degree $k$ and $m$ respectively) such that $\mathcal{A}_{(k+m)}=\mathcal{A}_{(k)}+\mathcal{A}_{(m)}$. This notation was firstly introduced by Novikov in \cite{Novi2}. 
As an example, the second Hamiltonian structure of the celebrated Korteweg-de Vries equation equation (namely the {Magri} operator) belongs to the class of $(3+1)$ non-homogeneous Hamiltonian operators, being $\mathcal{A}_{(3+1)}=\partial_x^3+2u\,\partial_x+u_x$.

A non-homogeneous operator $\mathcal{A}_{(k+m)}$ is  Hamiltonian if and only if $\mathcal{A}_{(k)}$ and $\mathcal{A}_{(m)}$ are both independently Hamiltonian, and the Schouten bracket between the operators of mixed orders $[[\mathcal{A}_{(k)},\mathcal{A}_{(m)}]]$ vanishes\footnote{The Schouten bracket is an 
extension of the commutator of vector fields on the space of local multi-vectors $\Lambda^k$. It is a bilinear operation $\Lambda^\ell \times \Lambda^k \to \Lambda^{\ell+k-1}$ such that it coincides with the commutator of local vector fields for $\ell = k = 1$ and satisfies a graded Leibniz property on the exterior product, i.e.
\begin{equation*}
    [[\mathcal{A}, \mathcal{X} \wedge \mathcal{Y}]] = [[\mathcal{A}, \mathcal{X}]] \wedge \mathcal{Y} + (-1)^{(a-1)x}\mathcal{X} \wedge [[\mathcal{A}, \mathcal{Y}]]
\end{equation*}
for $\mathcal{A} \in \Lambda^a$ and $\mathcal{X}\in \Lambda^x$.}. 

We can now identify the operator $\mathcal{A}_{(1+0)}$ in \eqref{eq:op_1_0} as the non-homogeneous operator of type~$(1+0)$ given by 
\begin{equation} \label{eq:operator_c}
  \mathcal{A}^{ij}_{(1+0)}= \mathcal{A}_{(1)}^{ij} + \mathcal{A}_{(0)}^{ij} = \big(g^{ij}(u)\,\partial_x+b^{ij}_k(u)\,u^k_x\big)+\big(\omega^{ij}(u)\big)\,. 
\end{equation}
Such operators were introduced in~\cite{DubrovinNovikov:PBHT}, as the simplest extension of first-order operators for hydrodynamic systems, for this reason also known as operators of \emph{non-homogeneous hydrodynamic type}.

The following theorem holds true.

\begin{theorem}[\cite{FerMok1,mokhov98:_sympl_poiss}]\label{thm1}
A non-homogeneous operator $\mathcal{A}_{(1+0)}$ of type $(1+0)$ is Hamiltonian if and only if
\begin{enumerate}[label=\roman*.]
    \item $\mathcal{A}_{(1)}=g^{i j}\, \partial_x + b^{i j}_{\,\, k} \, u^k_x$ 
is Hamiltonian,
\item $\mathcal{A}_{(0)}=\omega^{i j}$ is Hamiltonian, and
\item the compatibility conditions are satisfied
\begin{align}\label{cond1}
    \Phi^{i j k} &= \Phi^{k i j} \,,
\\\label{eq:phi}
        \frac{\partial \Phi^{i j k}}{\partial u^r} & = \sum_{(i,j,k)}  b^{s i}_{r}\, \frac{\partial \omega^{j k}}{\partial u^s} + \left( \frac{\partial b^{i j}_{r}}{\partial u^s} - \frac{\partial b^{i j}_{s}}{\partial u^r}  \right)\omega^{s k}\,,
\end{align}
where $\Phi^{i j k}$ is the $(3,0)$-tensor
\begin{align}
        \Phi^{i j k} & = g^{i s}\, \frac{\partial \omega^{j k}}{\partial u^s} - b^{i j}_{s} \, \omega^{s k} - b^{i k}_{s} \, \omega^{j s} \,. \end{align} 
\end{enumerate}
\end{theorem}
Lastly, we mention that in the non-degenerate case ($\det g^{ij} \neq 0$), the expression in~\eqref{eq:non_deg_system_hom_1} 
becomes
\begin{equation}
    u^i_t=\left(\nabla^i\,\nabla_j\, h(u)\right)u^j_x+ (\tilde{\nabla}^i \,h(u)), \qquad i=1,2,\dots n, 
\end{equation}
where $\tilde{\nabla}^i=\omega^{ij}\,\nabla_j$.

We conclude this section by presenting two paradigmatic examples of quasilinear systems and their associated Hamiltonian operators: the sinh-Gordon for $2$ components and the inverted KdV for $3$ components. 

\begin{example}
In the sinh-Gordon equation in the field variable $\varphi(\xi,\tau)$ 
\begin{equation}
    \phi_{\xi \tau} = \sinh \varphi
\end{equation}
the transformation $\varphi = 2 \log u$ and the introduction of $v=2 u_{\tau}/u$ yield the quasilinear system 
\begin{equation}
    \begin{cases}
        u_t = \dfrac{1}{2}\,uv \\[1ex]
        v_t = v_x + \dfrac{1}{2}\left(u^2  - \dfrac{1}{u^2} \right) 
    \end{cases}\,,
\end{equation}
in the light-cone coordinates $\tau = t$ and $\xi = t-x$. The Hamiltonian operator related to this system has the form $\mathcal{A}_{(1+0)}$, i.e.\
\begin{equation}
    \mathcal{A}_{(1+0)}^{ij} = \begin{pmatrix}
        0 & 0 \\[1ex]
        0 & 1
    \end{pmatrix}\partial_x + \frac{1}{2}\begin{pmatrix}
        0 & u \\[1ex]
        -u & 0
    \end{pmatrix}
\end{equation}
with $g^{ij}$ degenerate, $b^{ij}_k$ vanishing identically in all the components and $\omega^{ij}$ depending on the variable $u$ only. 
\end{example} 

A particularly interesting class of quasilinear systems is determined by applying the so-called \emph{inversion procedure} introduced in \cite{tsa3} by Tsarev to scalar evolutionary systems, i.e.\ of the form 
\begin{equation}\label{evs1}
    u_t=F(u,u_x,\dots , u_{(k-1)x})+G(u,u_x, \dots u_{(k-1)x})\, u_{kx},
\end{equation}
with $F,G$ smooth functions. Introducing the auxiliary variables
\begin{equation}
    u^1=u,\quad u^2=u_x, \quad \dots \quad  u^{k}=u_{(k-1)x},
\end{equation}
the scalar equation \eqref{evs1} is mapped into the $k$-component system 
\begin{equation}
     u^1_x=u^2,\quad 
        u^2_x=u^3,\quad 
        \cdots \quad  u^k_x=\dfrac{F(u^1,\dots u^k)}{G(u^1,\dots u^k)}u^1_t+\dfrac{F(u^1,\dots, u^k)}{G(u^1,\dots u^k)}\,,
\end{equation}
which is quasilinear and non-homogeneous. After applying the inversion of the independent variables $x$ and $t$ the obtained equations constitute a quasilinear system of the form~\eqref{sys}. The procedure of increasing the number of variables and equations by decreasing the degree of derivation was investigated by Tsarev in \cite{tsa3} and by two of the present authors in~\cite{DellAVer1}. 
\begin{example}[\cite{mokhov98:_sympl_poiss}]\label{ex:inverted_kdv}
    The inverted system for the KdV equation 
    \begin{equation} 
    u_t=6uu_x+u_{xxx}
    \end{equation}
after introducing the variables $u^1=u,u^2=u_x$ and $u^3=u_{xx}$ is
\begin{equation}\label{kdvsys}
\begin{cases}
u^1_t=u^2\\u^2_t=u^3\\u^3_t=u^1_x+6u^1u^2\end{cases}\,. 
\end{equation} 
In addition, in \cite{mk2} the author finds a transformation of variables 
{\small
\begin{equation}
    u^1=\frac{w^1-w^3}{\sqrt{2}},\qquad  u^2=w^2,\qquad u^3=\frac{w^1+w^3}{\sqrt{2}}+\left(w^1-w^3\right)^2 ,
\end{equation}}
also known as \emph{local quadratic unimodular change}, such that the KdV equation reads as 
{\small
\begin{equation}\label{kdvsys2}
    \begin{cases}
    w^1_t=-\dfrac{1}{2}\left(w^1-w^3\right)_x+w^2\left(w^1-w^3\right)+\dfrac{1}{\sqrt{2}}w^2\\
    w^2_t=\left(w^1-w^3\right)^2+\dfrac{1}{\sqrt{2}}\left(w^1+w^3\right)\\
    w^3_t=-\dfrac{1}{2}\left(w^1-w^3\right)_x+w^2\left(w^1-w^3\right)-\dfrac{1}{\sqrt{2}}w^2    \end{cases} \,.
\end{equation}}

\noindent
Despite this is not the most usual presentation of the KdV equation, it is possible to prove its integrability by introducing a bi-Hamiltonian structure with the following operators:
{\small
\begin{subequations}\label{kdvops}
\begin{align}  
&\mathcal{A}^{ij}_{(1+0)}=\begin{pmatrix}
1&0&0\\0&-1&0\\0&0&-1
\end{pmatrix}\partial_x+\begin{pmatrix}
0&-2w^3&2w^2\\2w^3&0&2w^1\\-2w^2&-2w^1&0
\end{pmatrix},\\[2ex]
&\mathcal{B}^{ij}_{(1+0)}=\frac{1}{2}\begin{pmatrix}
1&0&1\\0&0&0\\1&0&1
\end{pmatrix}\partial_x+\begin{pmatrix}
0&w^1-w^3+\frac{1}{\sqrt{2}}&0\\
w^3-w^1-\frac{1}{\sqrt{2}}&0&w^3-w^1+\frac{1}{\sqrt{2}}\\0&w^1-w^3-\frac{1}{\sqrt{2}}&0
\end{pmatrix},\label{kdvopsb}
\end{align}
\end{subequations}}
with $\mathcal{A}_{(1+0)}$ non-degenerate in its leading coefficient $g_{\mathcal{A}}(w)$ and $\mathcal{B}_{(1+0)}$ degenerate in~$g_{\mathcal{B}}(w)$. 

\end{example}
We refer to \cite{DellAVer1} for a classification of the Hamiltonian structures of the type $\mathcal{A}_{(1+0)}$ associated with quasilinear systems in $2$ and $3$ components, extending the results in \cite{Savoldi1} for degenerate operators of type $\mathcal{A}_{(1)}$ and to \cite{Riz1} for a recent classification of similar operators in more then one dimension.

\section{Casimirs of the operators}\label{sec3}
In this Section, we introduce the Casimir functionals associated with non-homogeneous operators $\mathcal{A}_{(1+0)}$, which will be relevant in the study of the bi-Hamiltonian systems treated in the next Section. We give the explicit form of the Casimirs both for the degenerate and the non-degenerate cases.

A Casimir $F=\int f \,dx$ of the Hamiltonian operator $\mathcal{A}$ is a functional such that in local coordinates
\begin{equation}\label{caseq}
    \mathcal{A}^{ij}\,\frac{\delta F}{\delta u^j}=0, \qquad i=1,2,\dots n.
\end{equation}
In terms of Poisson brackets, this is equivalent to requiring 
$\{F,G\}_{\mathcal{A}} = 0$ for any functional $G$.
For non-homogeneous hydrodynamic type operators $\mathcal{A}_{(1+0)}$ as in~\eqref{eq:operator_c}, formula~\eqref{caseq} reads as 
\begin{equation}
    \left(g^{ij}(u)\partial_x+b^{ij}_k(u)u^k_x+\omega^{ij}(u)\right)\frac{\delta F}{\delta u^j}=0, \qquad i=1,2,\dots n.
\end{equation}

\begin{remark}[On hydrodynamic functionals] \label{rem}
 A particular class of functionals $F$ is given by those whose densities $f$ only depend on the field variables and not on their derivatives. 
They play a key role in the theory of systems of hydrodynamic type, especially for bi-Hamiltonian systems with compatible Dubrovin-Novikov operators. This is the only possible choice for quasilinear systems.

\end{remark}
Because of Remark \ref{rem}, the hydrodynamic Casimir $F(u^1,\dots u^n)$ of operators~$(1+0)$ is solution of the following system:
\begin{subequations}
\begin{align}
\left(g^{ij}\partial_x+b^{ij}_ku^k_x\right)\frac{\partial F}{\partial u^j}=0, \qquad i=1,2,\dots n,\label{cas1}\\[1ex]
    \omega^{ij}\frac{\partial F}{\partial u^j}=0, \label{cas0}\qquad i=1,2,\dots n.
\end{align}
\end{subequations}
In the following two subsections, we distinguish two very different cases: operators with non-degenerate leading coefficients $g^{ij}$ and operators whose leading coefficient has no maximal rank.

\subsection{Non-degenerate case}
Let us now assume the non-degeneracy hypothesis on $\mathcal{A}_{(1+0)}$ in \eqref{eq:operator_c}, i.e.\ let $\det g\neq 0$. As we saw in the previous Section, the first-order terms in $\mathcal{A}_{(1)}$ define a Dubrovin-Novikov operator~\eqref{dnop_tens} satisfying the conditions \eqref{eq:cond_ham_nondeg_first_order} for which $b^{ij}_k=-g^{is}\,\Gamma^j_{sk}$. In this case, the Casimir equation~\eqref{cas1} reads as
\begin{equation}
  g^{ik}\frac{\partial F}{\partial u^k\partial u^j}+\Gamma^{ik}_j\frac{\partial F}{\partial u^k} = \nabla^i\, \nabla_j \, F=0, \qquad i,j=1,2,\dots n\,.
\end{equation}
From the geometric point of view, the solution $F$ of this equation can be written in terms of $n$ functionally independent densities $F=c_1F^1+\dots c_nF^n$ such that the change of coordinates
\begin{equation}\label{darbcoorcas}
    \tilde{u}^1=F^1,\quad \dots \quad \tilde{u}^n=F^n,
\end{equation}
reduces the leading coefficient $g$ to a constant form, and all the symbols $\Gamma^{ij}_k$ vanish identically. The constant form for  $\mathcal{A}_{(1)}$ is also known as Darboux form. 

In the recent work \cite{GubOliSgrVer}, the authors studied non-degenerate operators of type~$(1+0)$, in which the leading coefficient is in constant form and, consequently, the ultralocal operator is linear in the field variables. 
The following result holds true:
\begin{theorem}[\cite{GubOliSgrVer}]\label{gencasthm}
    In the non-degenerate case, operators of type $
(1+0)$ in Darboux form admit only linear Casimirs. 
\end{theorem}

\begin{example}Let us consider the first operator of the KdV in Example \ref{ex:inverted_kdv}. The operator $\mathcal{A}_{(1+0)}$ is non-degenerate and explicitly reads as
\begin{equation}
    \mathcal{A}^{ij}_{(1+0)}=\begin{pmatrix}
1&0&0\\0&-1&0\\0&0&-1
\end{pmatrix}\partial_x+\begin{pmatrix}
0&-2w&2v\\2w&0&2u\\-2v&-2u&0
\end{pmatrix}\,
\end{equation}
In accordance with Theorem \ref{gencasthm} the Casimirs of $\mathcal{A}_{(1)}$ are given by
\begin{equation}
    F_1(u,v,w)=c_1 u+ c_2 v+ c_3 w+ c_4,
\end{equation}
where $c_i$ are arbitrary constants. On the other hand, the Casimirs of $\mathcal{A}_{(0)}$ are arbitrary functions in $\xi=uw+v^2$, i.e.
\begin{equation}
    F_0(uw+v^2).
\end{equation}
Combining them the only possible Casimir for both is trivial, that is $F=c_4$.
\end{example}

A more general and interesting situation is covered by the operators with degenerate leading coefficients. 

\subsection{Degenerate case}\label{sec:casimir_degenerate_case}

The present investigation is inspired by the KdV bi-Hamiltonian structure as described in Example~\ref{ex:inverted_kdv} 
where the second operator $\mathcal{B}_{(1+0)}$ has degenerate leading coefficient of rank $2$. 
\begin{example}
The operator 
\begin{equation}\label{eq:operator_B_kdv}
    \mathcal{B}^{ij}_{(1+0)}=\frac{1}{2}\begin{pmatrix}
1&0&1\\0&0&0\\1&0&1
\end{pmatrix}\partial_x+\begin{pmatrix}
0&u-w+\frac{1}{\sqrt{2}}&0\\
w-u-\frac{1}{\sqrt{2}}&0&w-u+\frac{1}{\sqrt{2}}\\0&u-w-\frac{1}{\sqrt{2}}&0
\end{pmatrix}
\end{equation}
admits the most general Casimir function of the form
\begin{equation}\label{casb2}
     F(u,v,w)= (u-w)^2- \sqrt{2} (u+w). 
\end{equation}
The Casimir associated with 
$\mathcal{B}_{(1)}$ is 
\begin{equation}
     F_1(u,v,w)=c_1(u+v)+ \varphi(v)+\psi(v,u-w)\,,
\end{equation}
where $\varphi, \psi$ are arbitrary functions in the indicated variables, $c_1$ an arbitrary constant. The Casimir for the ultralocal term $\mathcal{B}_{(0)}$ is
\begin{equation}
  F_0(u,v,w)=\theta(\xi), \qquad    \xi=(u-w)^2- \sqrt{2} (u+w),
\end{equation}
i.e.\ any arbitrary function $\theta$ in the variable $\xi$. We emphasise that $F$ in \eqref{casb2} serves as the Hamiltonian density of the first structure $\mathcal{A}_{(1+0)}$ of the pair.
\end{example}

A general result for the Casimirs associated with degenerate operators is difficult to obtain. Indeed, we remark that no clear geometric (or covariant) interpretation of these operators is currently available in the literature, due to the impossibility to define a unique set of compatible symbols of connections for degenerate symmetric tensors. For this reason, we choose to present a complete list (with the most general solutions) of Casimirs for the classified degenerate operators of type $(1+0)$ in $n=2,3$ components. The classification here presented is then formulated up to diffeomorphisms of the field variables $(u^1,\dots, u^n)$. We use the classification results obtained in~\cite{DellAVer1} for degenerate operators $\mathcal{C}_{(1+0)}$ and we list their Casimir functions in Tables 1 and 2 for operators in $n=2$ and $n=3$ components respectively. According to the notation used in~\cite{DellAVer1} each degenerate operator $\mathcal{C}_{(1+0)}=C^{ij}_{\ell,k}$ is identified by two indices: with~$\ell$ we indicate the number of components and with $k$ the enumeration of the operators (listed in Appendix \ref{app:old_classification}). 

For the sake of clarity, we illustrate the procedure for two selected operators $\mathcal{C}_{(1+0)}$ from the classification in \cite{DellAVer1}, namely $C^{ij}_{3,2}$ and $C^{ij}_{3,5}$, that we recall here: 
{\small
\begin{subequations}
\begin{align}
C^{ij}_{3,2}&=
    \begin{pmatrix}
        1&0&0\\
        0&0&0\\
        0&0&0
    \end{pmatrix}\partial_x + 
    \begin{pmatrix}
        0&f(v,w)&g(v,w)\\
        -f(v,w)&0&h(v,w)\\
        -g(v,w)&-h(v,w)&0
    \end{pmatrix},  \label{eq:operator_C32} \\[2ex]
    C^{ij}_{3,5}&=
    \begin{pmatrix}
        1&0&0\\
        0&0&0\\
        0&0&0
    \end{pmatrix}\partial_x 
    + \frac{1}{u}  \begin{pmatrix}
        0 & -v_x & -w_x \\
        v_x & 0 & 0 \\
        w_x & 0 & 0 \\
    \end{pmatrix}
    + \frac{1}{u}
    \begin{pmatrix}
        0&f(v,w)&g(v,w)\\
        -f(v,w)&0&h(v,w)\\
        -g(v,w)&-h(v,w)&0
    \end{pmatrix}, \label{eq:operator_C35}
\end{align}
\end{subequations}}

\noindent
with $f$, $g$, $h$ arbitrary functions in $(v,w)$ satisfying the constraint
{\small
\begin{equation}\label{jacob1}
   f \! \left(v , w\right) \frac{\partial}{\partial v}h \! \left(v , w\right)-h\!\left(v , w\right) \frac{\partial}{\partial v}f \! \left(v , w\right)+\mathit{g} \! \left(v , w\right) \frac{\partial}{\partial w}h \! \left(v , w\right)-h \! \left(v , w\right) \frac{\partial}{\partial w}\mathit{g} \! \left(v , w\right)=0 \,, 
\end{equation}}

\noindent
given by the vanishing of the Jacobi identity and also referred to as closure relation in the following.

We first consider the Casimirs of the term $\mathcal{C}_{(1)}$ in $C^{ij}_{3,2}$ and $C^{ij}_{3,5}$ respectively, i.e.\ the first term on the right hand side of \eqref{eq:operator_C32} and the first two terms on the right hand side of  \eqref{eq:operator_C35}.
Solving equation \eqref{cas1}, we obtain for $\mathcal{C}_{(1)}$ in $C^{ij}_{3,2}$,
\begin{equation}
    F_1(u,v,w)=c\, u+ \varphi(v,w),
\end{equation} 

with $c$ constant and $\varphi$ and arbitrary function in $(v,w)$, whereas for $\mathcal{C}_{(1)}$ in $C^{ij}_{3,5}$ we have
$
    F_1(u,v,w)= c,
$
with $c$ a constant.

We now consider the part $\mathcal{C}_{(0)}$ in both operators, i.e.\ we compare the Poisson tensors~$\omega^{ij}_{3,2}$ and~$\omega^{ij}_{3,5}$ 
in \eqref{eq:operator_C32} and  
in \eqref{eq:operator_C35} respectively. We observe they are quite similar, and related by 
\begin{equation}\label{eq:ultralocal_35_32}
    \omega^{ij}_{3,5}=\frac{1}{u}\,\omega^{ij}_{3,2}. 
\end{equation}
Due to equation \eqref{cas0}, the Casimir functions of the part $\mathcal{C}_{(0)}$ in \eqref{eq:ultralocal_35_32} coincide. In particular, $F_{0}(u,v,w)$ is a Casimir for $\mathcal{C}_{0}$ if and only if the following system is satisfied
\begin{subequations}\label{sys43}
\begin{align}
   f(v,w) \, \dfrac{\partial}{\partial v}F_0(u,v,w)+g(v,w) \, \dfrac{\partial }{\partial w}F_0(u,v,w)&=0,  \label{sys43a}\\
    f(v,w) \, \dfrac{\partial }{\partial u}F_0(u,v,w)-h(v,w) \, \dfrac{\partial }{\partial w} F_0(u,v,w)&=0, \label{sys43b}\\
     g(v,w) \, \dfrac{\partial}{\partial u}F_0(u,v,w)+h(v,w) \, \dfrac{\partial }{\partial v}F_0(u,v,w)&=0. \label{sys43c}
\end{align}
\end{subequations}
The general solution to this system is hard to find, as the form of $F_0$ depend on the the functions $f(v,w),g(v,w)$ and $h(v,w)$, which in turn satisfy the further constraint \eqref{jacob1}. Nevertheless, it serves as a useful tool to verify whether a given function $F$ is a Casimir for the operator. In the following, we list the explicit solutions in all the subcases for $\mathcal{C}_{(0)}$ for both operators and~$\mathcal{C}_{(1+0)}$ for the operator $C^{ij}_{3,2}$ only. Indeed, the whole operator $\mathcal{C}_{(1+0)}$ for~$C^{ij}_{3,5}$ the resulting Casimir is trivial.  
\begin{enumerate}
{
    \item {Case $f(v,w)=0$}\\ 
    If $f(v,w)=0,$ $g(v,w) \neq 0$ and $h(v,w)\neq 0,$ the closure condition becomes
    \begin{equation}
      h(v,w) \,\dfrac{\partial}{\partial w} g(v,w)-   g(v,w) \,\dfrac{\partial}{\partial w} h(v,w) =h^2(v,w) \, \dfrac{\partial}{\partial w} \left(\dfrac{g(v,w)}{h(v,w)} \right)=0 \,.
    \end{equation}
    Hence,  ${g(v,w)}/{h(v,w)} \equiv \ell(v),$ with $\ell$ an arbitrary function.  
   Furthermore, from~\eqref{sys43} it immediately follows that $F_0(u,v,w)=F_0(u,v),$ so that the first two equations are verified. Substituting in the third one, we get
   
   \vspace*{-2ex}
   
    \begin{equation}
      \dfrac{\partial }{\partial u}F_0(u,v)+    \dfrac{h(v,w)}{g(v,w)} \dfrac{\partial }{\partial v}F_0(u,v)=0 \,.
    \end{equation}
    By integration, through the method of characteristics we obtain 
    \begin{equation}
        F_0(u,v)= f_0 \left(u- \int \! \dfrac{g(v,w)}{h(v,w)}\,dv \right)=f_0 \left(u-\int \ell(v) dv \right),
    \end{equation}
    where $f_0$ is an arbitrary function in its argument. By comparison with \eqref{cas1}, we get that the Casimir for the complete operator $\mathcal{C}_{(1+0)}$ is 
    \vspace*{-1ex}
     \begin{equation}
        F(u,v)= u- \int \ell(v) dv.
    \end{equation}
}
    \vspace*{-3ex}
    
    \item {Case $f(v,w)=g(v,w)=0$ }
    
  The closure is trivially satisfied and from \eqref{sys43} we get
  \vspace*{-1ex}
\begin{equation}
    \dfrac{\partial}{\partial w}F_0(u,v,w)= \dfrac{\partial}{\partial v}F_0(u,v,w)=0,
\end{equation}
so that the Casimir of the operator reduces to an arbitrary function of $u$. The Casimir for the full operator is $F=u.$
    \item {Case $f(v,w)=h(v,w)=0$}
    
       The closure is guaranteed and from \eqref{sys43} we have
       \vspace*{-1ex}
    \begin{equation}
         \dfrac{\partial}{\partial u}F_0(u,v,w)= \dfrac{\partial}{\partial w}F_0(u,v,w)=0,
    \end{equation}
    so that the Casimir of the operator reduces to an arbitrary function of the variable $v$ and it coincides with the Casimir of the operator $\mathcal{C}_{(1+0)}$.  

   \item {Case $g(v,w)=0,$ $f(v,w) \neq 0$ and $h(v,w)\neq 0$}  
    
    The corresponding Casimir is given by
    \begin{equation}
        F_0(u,w)= f_0 \left(u- \int \! \dfrac{f(v,w)}{h(v,w)}\,dw \right)= f_0 \left(u- \int \ell(w) dw \right),
    \end{equation}
      where $f_0$ is an arbitrary function in its argument. In this case,
      \vspace*{-1ex}
       \begin{equation}\label{eq:casimir_kdv}
        F= u- \int \ell(w) dw.
    \end{equation}

    \vspace*{-3ex}
    
         \item {Cases $g(v,w)=f(v,w)=0$ and $g(v,w)=h(v,w)=0$}
         
         The corresponding Casimir are given by an arbitrary function of the variable $u$ and an arbitrary function in $w$ respectively. Hence, the Casimir of the full operator is either given by $F=u$ or by $F=f(w)$, with $f$ arbitrary. 
       
\item {Case $h(v,w)=0,$ $f(v,w) \neq 0$ and $g(v,w)\neq 0$}\\
 The closure condition is trivially satisfied. The corresponding Casimir $F_0=f_0(v,w)$ and must satisfy the remaining equation $\eqref{sys43a}$, that gives also the Casimir of the full operator. We stress that the closed solution of this equation highly depend on the choice of the non-zero functions in the ultralocal term.
\end{enumerate}

\begin{table}
\centering
{\small
    \begin{tabular}{cccc}
    \hline \\[-1ex]
Operator&$\mathcal{C}_{(0)}$&$\mathcal{C}_{(1)}$&$\mathcal{C}_{(1+0)}$\\[1.8ex]\hline \\[-1ex]
      $C^{ij}_{2,1}$    &$c_1$&$c_1u+c_2\,h(v)$&$c_1$\\[1.8ex]\hline \\[-1ex]
       $C^{ij}_{2,2}$    &$c_1$&$c_1$&$c_1$\\[1.8ex]\hline \\[-1ex]
       \end{tabular}}
    \caption{Casimir classification of $(1+0)$ type operators in $n=2$ components}
    \label{tab:casimir_2}
\end{table}

\begin{example}[Generalised KdV equation]\label{ex3} Let us consider the generalised KdV equation
\begin{equation}\label{mkdv}
    u_t+3(n+1)u^nu_x+u_{xxx}=0
\end{equation}
where $n$ is a positive integer. We introduce the new variables $u=u, v=u_{x}$ and $w=u_{xx}$ so that after applying the inversion procedure, the equation reads as the quasilinear system 
\begin{equation}\label{sysmkdv}
    \begin{cases}
    u_t=v\\
    v_t=w\\
    w_t=-u_x-3(n+1)u^nv
    \end{cases}\,. 
\end{equation}
 As proved in \cite{DellAVer1}, the present system has Hamiltonian structure with the operator 
\begin{equation}\label{opmkdv}  
    \mathcal{C}^{ij}_{(1+0)}=\begin{pmatrix}
    0&0&0\\0&0&0\\0&0&1
    \end{pmatrix}\partial_x+\begin{pmatrix}
    0&1&0\\-1&0&-3(n+1)(u)^{n-1}\\
    0&3(n+1)(u)^{n-1}&0
    \end{pmatrix}
\end{equation} 
with Hamiltonian functional 
\begin{equation}
    H(u,v,w)=\int{\left(3(u)^{n+1}-uw+\frac{v^2}{2}\right) dx}.
\end{equation}
The operator \eqref{opmkdv} belongs to the class of $C^{ij}_{3,2}$ in \eqref{eq:operator_C32} for a particular choice of the arbitrary functions. The Casimir function associated with $\mathcal{C}_{(1+0)}$ in \eqref{opmkdv} is 
\begin{equation}
F(u,w)=c_1\!\left(w-3 \frac{u^n}{n}\right)+c_2,     
\end{equation}
with $c_1$, $c_2$ constants. It is easy to check that $F(u,w)$ is exactly given by expression \eqref{eq:casimir_kdv} for $\ell(w)=3w^{n-1}$, under the exchange of $u$ and $w$. Note that for $n=1$ the generalized KdV reduces to the KdV equation, and the Hamiltonian operator $\mathcal{C}_{(1+0)}$ in \eqref{opmkdv} reduces to $\mathcal{B}_{(1+0)}$ in \eqref{eq:operator_B_kdv}.   
\end{example}
Applying a similar procedure to the whole list of the previously classified operators $\mathcal{C}_{(1+0)}$, we obtain the general integral for all the cases (and subcases) arising. 
In Table~\ref{tab:casimir_2} and Table~\ref{tab:casimir_3} we list the Casimir computed for $n=2$ (i.e.\ $(u,v)$) and the remaining for $n=3$ (i.e.\ $(u,v,w)$) respectively. In both tables, 
$c_1,c_2,c_3$ are $c_4$ are arbitrary constants, $\varphi,\phi$ are arbitrary functions in their arguments, and $f, g, h$ are functions fixed and uniquely determined by the corresponding operator.

\newgeometry{left=2cm,right=3cm,top=2cm}
\begin{table}[]
    \centering
    {\small
    \begin{tabular}{ccccc}
                \hline \\[-1ex]
Operator&Case&$\mathcal{C}_{(0)}$&$\mathcal{C}_{(1)}$&$\mathcal{C}_{(1+0)}$\\[2ex]\hline \\[-1ex]
      $C^{ij}_{3,1}$    &&$\varphi(w)$&$\varphi(w)$&$\varphi(w)$\\[2ex]\hline \\[-1ex]
     $C^{ij}_{3,3}$&&     $\varphi(w)$&$\varphi(w)$&$\varphi(w)$\\[2ex]\hline \\[-1ex]
        $C^{ij}_{3,4}$   &&$\varphi(v)$&$\varphi(v)$&$\varphi(v)$\\[2ex]\hline \\[-1ex]
                \multirow{2}{*}{$C^{ij}_{3,6}$}&$g(w)=0$&$\varphi(w)$&$c_1u+c_2v+\phi(w)$&$\varphi(w)$\\[2ex]\cline{2-5} \\[-.5ex]
            &otherwise&$ \phi\!\left(c v-{\displaystyle \int \! \dfrac{f(w)}{g(w)}dw} -u \right)$    &$c_1u+c_2v+\phi(w)$&$c v-{\displaystyle \int \! \dfrac{f(w)}{g(w)}dw} -u $\\[3ex]\hline \\[-4ex]
                   \rule[-0.45cm]{0mm}{1.2cm}
           \multirow{2}{*}{$C^{ij}_{3,7}$}&$c=0$&$\varphi(u)$&$ u$ & $ u$
\\[2ex]\cline{2-5} \rule[-0.45cm]{0mm}{1.2cm}
                                &$c\neq 0$&$\varphi\!\left(\dfrac{c(u^2+v^2)-2u}{c}\right)$&$u $&$c_2$\\[2ex]\hline \\[-2ex]
 \rule[-0.45cm]{0mm}{1.2cm}
           \multirow{2}{*}{$C^{ij}_{3,8}$}&$c=0$&$\varphi\!\left(\dfrac{ u+ vw }{\sqrt{w^{2}+1}}
\right)$&$\dfrac{c_{1} \left(u+ vw\right) }{\sqrt{w^{2}+1}}
$&$\dfrac{ c_{1}\left(u+ vw\right) }{\sqrt{w^{2}+1}}$\\[2ex] \cline{2-5} \\[-2ex] \rule[-0.45cm]{0mm}{1.2cm}
                                &$c\neq 0$&$\varphi\!\left(\dfrac{2\left( u+ vw \right)}{\sqrt{w^{2}+1}} -{c}(u^2+v^2)
\right)$&$\dfrac{c_{1}\left(u+ vw\right)  }{\sqrt{w^{2}+1}}$&$c_3$\\[3ex]\hline \\[-1ex]
\multirow{2}{*}{$C^{ij}_{3,9}$}&$g(w)=0$&$\varphi(w)$&$c_1u+c_2v+\phi(w)$&$\varphi(w)$\\[2ex]  \cline{2-5}  \\[-1ex]                            &otherwise&$\phi\!\left(c v-{\displaystyle \int \! \dfrac{f(w)}{g(w)}dw} -u \right)$&$c_1u+c_2v+\phi(w)$&$c v-{\displaystyle \int \! \dfrac{f(w)}{g(w)}dw} -u $\\[3ex]\hline \\[-1ex]  \multirow{5}{*}{$C^{ij}_{3,10}$}&$f(w)=g(w)=0$&$\varphi\!\left(v\right)$&$\phi(w)v$&$v$\\[2ex]  \cline{2-5}  \\[-3ex]         
\rule[-0.45cm]{0mm}{1.2cm}&$f(w)=h(w)=0$&$\varphi(uv)$&$\phi(w)v$&$c_2$\\[2ex]  \cline{2-5}  \\[-1ex]         
                  & $g(w)=h(w)=0$& $\varphi\!\left(w\right)$ &$\phi(w)v$&$c_2$\\[2ex]  \cline{2-5}  \\[-3ex]
                    \rule[-0.45cm]{0mm}{1.2cm}&$g(w)=0$&$\varphi\!\left(v\,\text{exp}\left({- {\displaystyle\int{\!\dfrac{f(w)}{h(w)}\, dw}}}\right)\right)$&$\phi(w)v$&$c_2$\\[3ex]  \cline{2-5}  \\[-3ex]         \rule[-0.45cm]{0mm}{1.2cm} &otherwise &$\varphi\!\left(v\,\dfrac{ug(w)-h(w)}{g(w)}\right)$&$\phi(w)v$&$c_2$\\[3ex]\hline \\[-3ex]\rule[-0.45cm]{0mm}{1.2cm} $C^{ij}_{3,11}$
&&$\varphi\!\left(2cu\sqrt{w}+\dfrac{2vc}{\sqrt{w}}-uv\right)$&$\dfrac{ c_1( u+v w)}{\sqrt{w}}
$&$c_3$\\[3ex]\hline \\[-1ex]
    \end{tabular}}
    \caption{Casimir classification of $(1+0)$ type operators in $n=3$ components}
    \label{tab:casimir_3}
\end{table}

\restoregeometry

\section{Bi-Hamiltonian structures with non-homogeneous operators}\label{sec4}

The bi-Hamiltonian formalism plays a key role in integrable systems, as firstly shown by Magri in~\cite{Magri:SMInHEq} and then further investigated in several papers (see e.g.\ \cite{bla,Li,LSV:bi_hamil_kdv,lorvit,PavVerVit1}). We refer to~\cite{Olver:ApLGDEq} and \cite{DubZha} for an extensive treatment of this formal approach. Here we briefly recall that an evolutionary system is said to be bi-Hamiltonian if there exist two compatible Hamiltonian operators $\mathcal{A}$ and $\mathcal{B}$ such that the system can be written as 
\begin{equation}\label{eq:bi_ham_system}
    u^i_t=\mathcal{A}^{ij}\,\frac{\delta H_0}{\delta u^j}=\mathcal{B}^{ij}\,\frac{\delta H_1}{\delta u^j}, \qquad i=1,2,\dots , n,
\end{equation}
with two Hamiltonian functionals $H_0$, $H_1$. The two operators are compatible if any linear combination $\mu\,\mathcal{A}+\lambda\, \mathcal{B}$ is still a Hamiltonian operator. 

For bi-Hamiltonian systems, any functional $F$ representing a conserved quantity of the system gives rise to two Hamiltonian vector fields (i.e.\ $\{\,\cdot\,\,,\,F\,\}_{\mathcal{A}}$ and $\{\,\cdot\,\,,\,F\,\}_{\mathcal{B}}$). 
Both the Hamiltonian functionals $H_0$ and $H_1$ are conserved quantities, hence the vector field for $H_1$ is Hamiltonian also with respect to the operator $\mathcal{A}$. Hence, given \eqref{eq:bi_ham_system}, there exists a further functional $H_2$ such that 
\begin{equation}
    \mathcal{A}^{ij}\,\frac{\delta H_1}{\delta u^j}=\mathcal{B}^{ij}\,\frac{\delta H_2}{\delta u^j}.
\end{equation}
Iterating this procedure, an infinite sequence $\{H_i\}_{i \ge 0}$ of Hamiltonian functionals  emerges and to each symmetry $\{\,\cdot\,\,,\,H_j\}$ it can be associated a Hamiltonian functional $H_i$ such that 
\begin{equation}
    \{H_i,H_j\}_\mathcal{A}=\{H_i,H_j\}_\mathcal{B}=0\,, \qquad i \neq j\,. 
\end{equation}

With the additional requirement that the Hamiltonians are independent, 
 the evolutionary bi-Hamiltonian system~\eqref{eq:bi_ham_system} admits an infinite number of conserved quantities in involution, i.e.\ it is integrable. 

For our purposes, as described in the previous sections, the most natural structure for quasilinear systems is defined by $(1+0)$ operators. Therefore, here we focus on pairs of Hamiltonian operators\footnote{We refer to Theorem \ref{thm1} for the Hamiltonian property of $(1+0)$ operators.} $\mathcal{A}$ and $\mathcal{B}$ both of non-homogeneous hydrodynamic type, i.e.\ 
\begin{equation}\label{pairs}
\mathcal{A}^{ij}=g^{ij}_\mathcal{A}\,\partial_x+b^{ij}_{\mathcal{A},k}\,u^k_x+\omega^{ij}_\mathcal{A}\,,\qquad \mathcal{B}^{ij}=g^{ij}_{\mathcal{B}}\,\partial_x+b^{ij}_{\mathcal{B},k}\,u^k_x+\omega^{ij}_\mathcal{B},
\end{equation}
where the coefficients $g^{ij}_\mathcal{I},b^{ij}_{\mathcal{I},k}$ and $\omega^{ij}_\mathcal{I}$ (with $\mathcal{I}=\{\mathcal{A},\mathcal{B}\}$) depend on the field variables $u^k$ only~($k =1,\dots, n$). In a linear combination $\mu \, \mathcal{A}+\lambda \, \mathcal{B}$, we can assume one between $\lambda$ and $\mu$ to be nonzero. This allows us to study the expression
\begin{equation}\label{eq:non_homog_operators_AB}
\mathcal{A}^{ij}+\lambda \,\mathcal{B}^{ij}=\left(g^{ij}_\mathcal{A}+\lambda \,g^{ij}_\mathcal{B}\right)\partial_x+\left(b^{ij}_{\mathcal{A},k}+\lambda \,b^{ij}_{\mathcal{B},k}\right)u^k_x+\left(\omega^{ij}_\mathcal{A}+\lambda \, \omega^{ij}_\mathcal{B}\right),
\end{equation} 
that is again a non-homogeneous operator of type $(1+0)$. The fact that the resulting linear combination preserves the type of the composing operators will be a crucial point in our investigation. For the operator~\eqref{eq:non_homog_operators_AB} to be Hamiltonian, 
the following theorem holds true.
 \begin{theorem}\label{darv}
    Let $\mathcal{A}$ and $\mathcal{B}$ be two non-homogeneous Hamiltonian operators of type $(1+0)$ as in~\eqref{pairs}.  
    The operators $\mathcal{A}$ and $\mathcal{B}$ are compatible if and only if $ \mathcal{A}_{(1)}$ and $\mathcal{B}_{(1)}$ are compatible, and the following tensors identically vanish:
    {\small
    \begin{subequations}
    \begin{align}
    \label{cond0a}  
     L^{ijk}&=\frac{\partial \omega^{ij} _{\mathcal{A}}}{\partial u^p}\,\omega^{pk}_\mathcal{B}+\frac{\partial\omega^{ij}_{\mathcal{B}}}{\partial u^p}\,\omega^{pk}_\mathcal{A}+\frac{\partial \omega^{jk}_{\mathcal{A}}}{\partial u^p}\,\omega^{pi}_\mathcal{B}+\frac{\partial \omega^{jk}_{\mathcal{B}}}{\partial u^p}\,\omega^{pi}_\mathcal{A}+\frac{\partial \omega^{ki}_{\mathcal{A}}}{\partial u^p}\,\omega^{pj}_\mathcal{B}+\frac{\partial \omega^{ki}_{\mathcal{B}}}{\partial u^p}\,\omega^{pj}_\mathcal{A}\,,  \\[1ex]
    \begin{split}
     \label{cond1a}   
     P^{i j k} & = {g}^{i s}_\mathcal{A}\,  \frac{\partial \omega^{j k}_{\mathcal{B}}}{\partial u^s}- \frac{\partial g^{i j}_{\mathcal{A}}}{\partial u^s}\, {\omega}^{s k}_\mathcal{B} -{b}^{i k}_{\mathcal{A},s} \, {\omega}^{j s}_\mathcal{B} 
        +{g}^{i s}_\mathcal{B}\,  \frac{\partial \omega^{j k}_{\mathcal{A}}}{\partial u^s} - \frac{\partial g^{i j}_{\mathcal{B}}}{\partial u^s}\, {\omega}^{s k}_\mathcal{A} -{b}^{i k}_{\mathcal{B},s} \, {\omega}^{j s}_\mathcal{A}  \\[1ex]
        &~~+
         {g}^{j s}_\mathcal{A}\, \frac{\partial \omega^{i k}_{\mathcal{B}}}{\partial u^s} -{b}^{j k}_{\mathcal{A},s} \, {\omega}^{i s}_\mathcal{B} 
        +{g}^{j s}_\mathcal{B}\, \frac{\partial \omega^{i k}_{\mathcal{A}}}{\partial u^s} -{b}^{j k}_{\mathcal{B},s} \, {\omega}^{i s}_\mathcal{A} \, , \end{split}\\[1ex]
  \label{cond3a}  
  \begin{split} 
 S^{ijk}_r&=-g_{\mathcal{A}}^{is} \dfrac{\partial^2\omega_{\mathcal{B}}^{jk}}{\partial u^{s} \partial u^{r}}-g_{\mathcal{B}}^{is} \, \dfrac{\partial^2\omega_{\mathcal{A}}^{jk}}{\partial u^{s} \partial u^{r}}-\left(b_{\mathcal{A},r}^{is}+b_{\mathcal{A},r}^{si}\right) \dfrac{\partial \omega_{\mathcal{B}}^{jk}}{\partial u^{s}}-\left(b_{\mathcal{B},r}^{is}+b_{\mathcal{B},r}^{si}\right) \dfrac{\partial \omega_{\mathcal{A}}^{jk}}{\partial u^{s}}\\[1ex]
  &~~+ b_{\mathcal{A},s}^{ij} \,\dfrac{\partial \omega_{\mathcal{B}}^{sk}}{\partial u^{r}}
  + b_{\mathcal{A},s}^{ik} \,\dfrac{\partial \omega_{\mathcal{B}}^{js}}{\partial u^{r}}
  +b_{\mathcal{B},s}^{ij} \,\dfrac{\partial \omega_{\mathcal{A}}^{sk}}{\partial u^{r}}
  + b_{\mathcal{B},s}^{ik} \,\dfrac{\partial \omega_{\mathcal{A}}^{js}}{\partial u^{r}}\\[1ex]
  &~~+ \dfrac{\partial b_{\mathcal{A},s}^{ij}}{\partial u^{r}} \,\omega_{\mathcal{B}}^{sk} 
  +  \dfrac{\partial b_{\mathcal{A},s}^{ik}}{\partial u^{r}} \,\omega_{\mathcal{B}}^{js}
  +  \dfrac{\partial b_{\mathcal{B},s}^{ij}}{\partial u^{r}} \,\omega_{\mathcal{A}}^{sk}
  +  \dfrac{\partial b_{\mathcal{B},s}^{ik}}{\partial u^{r}} \,\omega_{\mathcal{A}}^{js}\\[1ex]
  &~~+\sum_{(i, \, j, \, k)} \left[ b_{\mathcal{A},r}^{si} \,\frac{\partial \omega_{\mathcal{B}}}{\partial u^s}^{jk}+b_{\mathcal{B},r}^{si} \,\frac{\partial \omega^{jk}_{\mathcal{A}}}{\partial u^s}+\left( \dfrac{\partial b_{\mathcal{A},r}^{ij}}{ \partial u^s} -\dfrac{\partial b_{\mathcal{A},s}^{ij}}{ \partial u^r} \right)w_{\mathcal{B}}^{sk}+\left( \dfrac{\partial b_{\mathcal{B},r}^{ij}}{ \partial u^s} -\dfrac{\partial b_{\mathcal{B},s}^{ij}}{ \partial u^r} \right)w_{\mathcal{A}}^{sk} \right],
  \end{split} 
\end{align}
\end{subequations} }

\noindent 
where $L=[[\omega_\mathcal{A},\omega_{\mathcal{B}}]]$ is the Schouten brackets for the ultralocal structures $\omega_\mathcal{A}$ and $\omega_\mathcal{B}$.
\end{theorem}
\begin{proof}
    The linear combination \eqref{eq:non_homog_operators_AB} 
    is a non-homogeneous operator of type $(1+0)$. Therefore, the coefficients defined as 
    \begin{align}\label{sorp}
&\tilde{g}^{ij}=g_\mathcal{A}^{ij}+ \lambda \,g^{ij}_\mathcal{B},\qquad \tilde{b}^{ij}_k= b^{ij}_{\mathcal{A},k}+\lambda \, b^{ij}_{\mathcal{B},k},\qquad \tilde{\omega}^{ij}=\omega^{ij}_\mathcal{A}+\lambda \,\omega^{ij}_\mathcal{B},
\end{align}
must satisfy Theorem \ref{thm1}. First, we notice that $\tilde{g}$ and $\tilde{b}$ must be such that the operator $\tilde{g}^{ij}\,\partial_x + \tilde{b}^{ij}_k  u^k_x $ is a Hamiltonian homogeneous operator, hence $(\mathcal{A}_{(1)},\mathcal{B}_{(1)})$ must be a compatible pair of Dubrovin-Novikov operators (part $i.$ of Theorem~\ref{thm1}). Concerning $\tilde{\omega}$ (part $ii.$ of Theorem~\ref{thm1}), the skew-symmetry is ensured for any choice of $\lambda$, whilst the Jacobi identity is non-trivial. The Jacobi identity for $\tilde{\omega}$ results in a polynomial in $\lambda$, for which the only non-zero term is the one linear in $\lambda$ and given by $L^{ijk}$ in~\eqref{cond0a}.

Lastly, we consider part $iii.$ of Theorem~\ref{thm1}. 
The condition \eqref{cond1} for \eqref{sorp} gives a polynomial of degree 2 in $\lambda$. One can easily see that the coefficients of $\lambda^0$ and $\lambda^2$ annihilate under the Hamiltonian property of $\mathcal{A}$ and $\mathcal{B}$ respectively. The term linear in $\lambda$ is given by $P^{ijk}$ in~\eqref{cond1a}. The final condition~\eqref{eq:phi} for~\eqref{sorp} translates to the tensor~$S^{ijk}_r$ in~\eqref{cond3a}. 
\end{proof}

We can interpret this result as follows. Given two Hamiltonian operators of type $(1+0)$, they are compatible if and only if their first-order and zero-order are, and the tensors $P^{ijk}$ and $S^{ijk}_r$ identically vanish. This implies that not all the pencils $(\mathcal{A}_{(1)},\mathcal{B}_{(1)})$ and $(\mathcal{A}_{(0)},\mathcal{B}_{(0)})$ generate the pencil~$(\mathcal{A},\mathcal{B})$, but they must form a more complicated geometric object that we call a bi-pencil, as we will deepen in section~\ref{bipen}. 

In the following, we present the general results on the pairs~\eqref{eq:non_homog_operators_AB} under the hypothesis that the operator $\mathcal{A}$ is non-degenerate. This assumption does not infer  any degeneracy property on~$\mathcal{B}$, hence the latter is completely general and we determine the form of its coefficients by requiring that~$(\mathcal{A},\mathcal{B})$ is a pair of compatible Hamiltonian operators.

\subsubsection*{The operator $\mathcal{A}$} 

Let us assume $\mathcal{A}$ to be Hamiltonian and non-degenerate. 
Then, there exists a local change of the dependent variables $u^m$ such that the flat metric $g^{ij}_\mathcal{A}$ is constant and diagonal, i.e. 
 \begin{equation}
     g^{ij}_\mathcal{A}=\eta^{ij}\,\delta_j^i,
 \end{equation}
 with $\eta^{ij}\in\mathbb{R}$. In the chosen set of Darboux coordinates, 
 the related Christoffel symbols identically vanish ($b^{ij}_{\mathcal{A},k}=0$), and the $(1+0)$ operator reduces to 
 \begin{equation}\label{eq:A_flat}
    \mathcal{A}^{ij}=\eta^{ij}\,\delta^i_j\,\partial_x + \omega^{ij}_\mathcal{A}\,.
 \end{equation}
With $\mathcal{A}$ non-degenerate, the Theorem \ref{thm1} gives rise to the following
 \begin{corollary}[\cite{FerMok1}]
     A non-degenerate non-homogeneous operator $\mathcal{A}$ of type $(1+0)$ is Hamiltonian if and only if $\mathcal{A}_{(1)}$ is Hamiltonian, $\omega_\mathcal{A}$ is a Poisson tensor and
     \begin{equation}
         \label{nondegcom}\nabla^i_\mathcal{A}\, \omega^{jk}_\mathcal{A}+\nabla_\mathcal{A}^{j}\, \omega^{ik}_\mathcal{A}=0\,, \qquad i,j,k=1,2,\dots n\,,
     \end{equation}
     where $\nabla_\mathcal{A}^i=g^{is}_\mathcal{A}\,\,\nabla_{\!\!\mathcal{A},s}^{\,}$ and $\nabla_{\!\!\mathcal{A},s}^{\,}$ is the covariant derivative with respect to the metric $g_{\mathcal{A}}$.
 \end{corollary}

 In particular, in flat coordinates the condition~\eqref{nondegcom} reads as
 \begin{equation}
    \eta^{is}\,\frac{\partial \omega^{jk}_{\mathcal{A}}}{\partial u^s}
    +\eta^{js}\,\frac{\partial \omega^{ik}_{\mathcal{A}}}{\partial u^s}=0\,, \quad \implies \quad \begin{cases}
        2\eta^{is}\,\dfrac{\partial \omega^{ik}_{\mathcal{A}}}{\partial u^s}=0 \,, \quad &i = j \\[1.5ex]
        \omega^{ij}_\mathcal{A}=c^{ij}_k\,u^k+f^{ij} \,, \quad &i \neq j
    \end{cases}
 \end{equation}
 with $c^{ij}_k,f^{ij}\in\mathbb{R}$ constants. Finally, requiring $\omega_\mathcal{A}$ to be a Poisson tensor we generate the further conditions
 \begin{subequations} 
 \begin{align}
    c^{ij}_sc^{sk}_l+c^{jk}_sc^{si}_l+c^{ki}_sc^{sj}_l&=0\,, \label{eq:c_constraints}\\
    f^{is}c^{jk}_s+f^{js}c^{ki}_s+f^{ks}c^{ij}_s&=0\,.
    \label{eq:c_f_constraints}
\end{align} 
\end{subequations} 
The operator $\mathcal{A}$ in Darboux coordinates is then given by the following 
\begin{equation}\label{eq:A_darboux}
\mathcal{A}^{ij}=\eta^{is}\delta_s^j \, \partial_x + (c^{ij}_k u^k+f^{ij})\,,
\end{equation}
and the constraints in~\eqref{eq:c_constraints} and~\eqref{eq:c_f_constraints}.

A complete description of non-homogeneous Hamiltonian operators in Darboux form has recently been developed by G. Gubbiotti, F. Oliveri, E. Sgroi and one of the present authors in \cite{GubOliSgrVer}. We refer to this paper for the general structure of such operators in flat coordinates up to $n=6$ number of components, and for further generalisations. In addition, we emphasise that formulas \eqref{eq:c_constraints} and \eqref{eq:c_f_constraints} have an interesting geometric interpretation~(e.g.\ \cite{mokhov98:_sympl_poiss,GubOliSgrVer}): $c^{ij}_k$ are the structure constants of a real Lie algebra with respect to which $f^{ij}$ is a $2$-cocycle. As described in \cite{GubOliSgrVer}, using Darboux coordinates for both the operators is equivalent to choosing an abelian Lie algebra endowed with a compatible scalar product. 

\begin{remark}[Darboux coordinates for $\mathcal{A}$]\label{darboux}
From the geometric point of view, the problem of identifying a change of variables for $\mathcal{A}$ that brings the operator into total constant form is equivalent to finding a Darboux transformation for both the metric $g$ and the Poisson tensor $\omega$ simultaneously. If additionally $g$ and $\omega$ are both non-degenerate, this is equivalent to finding the Darboux coordinates for $g$, and the symplectic form $\omega^{-1}$. In \cite{bdmt}, this problem was investigated for tensors of type $g+\omega$ with lower indices (i.e.\ with a covariant metric and a symplectic structure). In particular, the authors established that such a solution exists if and only if both the metric and the symplectic form are parallel via the existence of a covariant connection $\nabla$. Note that this type of transformation rule always exists for Dubrovin-Novikov operators,  being $g^{ij}$ the inverse of a covariant metric. 
\end{remark}

\subsubsection*{The operator $\mathcal{B}$} 

Once the operator $\mathcal{A}$ is set in Darboux form, we investigate the structure of the operator $\mathcal{B}$ by requiring compatibility. Following Theorem~\ref{darv}, this problem relates to the topic of compatible pairs of first-order homogeneous operators, that has been extensively studied (e.g.\ in \cite{Mok11,Mok12,Mok13,Pav1,Fer11}).  As far as the authors know the latest paper on this subject is the seminal work \cite{Mok13} by Mokhov, where the author mainly focuses on pairs of first-order operators whose leading coefficients are non-degenerate. We will discuss this aspect in Section \ref{bipen}.

In \cite{omokh}, Mokhov presented a necessary condition for the operator $\mathcal{B}_{(1)}$ to be compatible with the first-order operator $\mathcal{A}_{(1)}=\eta^{ij}\,\partial_x$. 

\begin{theorem}[Lemma 2, \cite{omokh}]\label{thm_mok}
    Any local Poisson structure of Dubrovin-Novikov type $\mathcal{B}_{(1)}$ is compatible with the operator $\mathcal{A}_{(1)}$ if and only if there exist  $h^1(\textbf{u}),\dots, h^n(\textbf{u})$ locally such that 
    \begin{equation}
        \label{opmok}
   \mathcal{B}_{(1)}= \left(\eta^{is}\,\frac{\partial h^j}{\partial u^s}+\eta^{js}\,\frac{\partial h^i}{\partial u^s}\right)\partial_x+\eta^{is}\,\frac{\partial^2h^j}{\partial u^s\,\partial u^k}\,u^k_x\,.
\end{equation}
\end{theorem}
This result is particularly useful for our purposes, since it allows a simplification of the problem by reducing the number of unknown functions in the metric and in the ultralocal term. From $n^2$ a priori independent entries ($(n^2+n)/2$ for the metric and $(n^2-n)/2$ for the ultralocal term) to $(n^2+n)/2$. We stress again the fact that there are no assumptions on the degeneracy property of $\mathcal{B}_{(1)}$ (and hence on $\mathcal{B}$).

It is important to notice that Mokhov's result in Theorem \ref{thm_mok} does not imply that $\mathcal{B}_{(1)}$ is automatically Hamiltonian. Therefore, we must impose that Theorem \ref{th:ham_A} is satisfied also for this operator.

With the further specification on the metric, Theorem \ref{darv} becomes:

\begin{corollary}\label{cormok}
If $\mathcal{A}$ is in flat coordinates, then $\mathcal{A}$ and $\mathcal{B}$ are compatible if and only if the following tensors annihilate: 
{\small
\begin{subequations}
\begin{align}
\label{eq:cond_G}
G^{ijk}&=\left(\eta^{ip}\dfrac{\partial h^\ell}{\partial u^p}+\,\eta^{\ell p}\, \dfrac{\partial h^i}{\partial u^p} \right)\eta^{js}\dfrac{\partial^2 h^k}{\partial u^s \partial u^\ell}-\left(\eta^{jp}\dfrac{\partial h^\ell}{\partial u^p}+\,\eta^{\ell p}\, \dfrac{\partial h^j}{\partial u^p} \right)\eta^{is}\dfrac{\partial^2 h^k}{\partial u^s \partial u^\ell}\,,\\[2ex]
R^{jr}_{sk}&=\,\dfrac{\partial^2 h^j}{\partial u^s \partial u^\ell}\,\eta^{\ell m}\, \dfrac{\partial^2 h^r}{\partial u^m \partial u^k}- \dfrac{\partial^2 h^r}{\partial u^s \partial u^\ell} \,\eta^{\ell m}\,\dfrac{\partial^2 h^j}{\partial u^m \partial u^k}\,,
\label{r1}\\[2ex]
\label{eq:cond_L}
\begin{split} 
L^{ijk}&=\,c_p^{ij}\,\omega^{pk}_\mathcal{B}+\frac{\partial \omega^{ij}_{\mathcal{B}}}{\partial u^p}(c^{pk}_su^s+f^{pk})
+c^{jk}_p\,\omega^{pi}_\mathcal{B}
+\frac{\partial \omega^{jk}_{\mathcal{B}}}{\partial u^p} 
(c^{pi}_su^s+f^{pi})    \\[1ex]&~~
+c^{ki}_p\,\omega^{pj}_\mathcal{B}  +\frac{\partial \omega^{ki}_{\mathcal{B}}}{\partial u^p}(c^{pj}_su^s+f^{pj}),
\end{split} 
\\[2ex] 
\label{eq:cond_P}
 \begin{split}
    {P^{ijk}}&=\,\eta^{is} \,\frac{\partial \omega^{jk}_{\mathcal{B}}}{\partial u^s}
    +\left(\eta^{i\ell}\,\frac{\partial h^s}{\partial u^\ell}+\eta^{s\ell}\,\frac{\partial h^i}{\partial u^\ell}\right) c^{jk}_s 
    -\eta^{i\ell}\,\frac{\partial^2 h^j}{\partial u^\ell \partial u^s} (c^{sk}_mu^m+f^{sk}) \\[1ex]
    &~~-\eta^{i\ell}\,\frac{\partial^2 h^k}{\partial u^\ell \partial u^s}(c^{js}_mu^m+f^{js})+\eta^{js} \dfrac{\partial \omega_{\mathcal{B}}}{\partial u^s}^{ik}+\left(\eta^{j\ell}\dfrac{\partial h^s}{\partial u^\ell}+\eta^{s\ell}\dfrac{\partial h^j}{\partial u^\ell}\right)c^{ik}_s \\[1ex]
    &~~-\eta^{j\ell}\dfrac{\partial^2 h^i}{\partial u^\ell \partial u^s}(c^{sk}_mu^m+f^{sk})-\eta^{j\ell}\dfrac{\partial^2 h^k}{\partial u^\ell \partial u^s}(c^{is}_mu^m+f^{is})\,,
    \end{split}\\[2ex]
    \label{eq:cond_S}
    \begin{split} 
     S^{jk}_{ir}&= 
      \dfrac{\partial^{2} \omega_{\mathcal{B}}^{jk}}{\partial u^{i} \partial u^{r}} 
     +  \dfrac{\partial^{2} h^{s}}{\partial u^{i } \partial u^{r}}
     c_{s}^{jk}
     +  \dfrac{\partial^{2} h^{k}}{\partial u^{i } \partial u^{s}} c_{r}^{js} +  \frac{\partial ^2 h^{k}}{\partial u^{s} \partial u^r}c_{i}^{s j}- \frac{\partial ^2 h^{j}}{\partial u^{s} \partial u^r}c_{i}^{s k} - \dfrac{\partial^{2} h^{j}}{\partial u^{i } \partial u^{s}} c_{r}^{sk}\\[1ex]
     &~~~~ -\dfrac{\partial^3 h^{j}}{\partial u^{i } \partial u^{s} \partial u^{r}} (c_{m}^{sk} u^{m}+ f^{sk})- \dfrac{\partial^3 h^{k}}{\partial u^{i } \partial u^{s} \partial u^{r} }(c_{m}^{js} u^{m}+ f^{js})\,.
     \end{split}
 \end{align}
 \end{subequations}}
\end{corollary}

\begin{proof}
We first observe that by Mokhov's Theorem it follows that the first-order operators $\mathcal{A}_{(1)}$ and~$\mathcal{B}_{(1)}$ form a compatible pair, i.e.\ no additional conditions need to be satisfied.

We now consider Theorem \ref{thm1} on $\mathcal{B}_{(1)}$. Conditions~\eqref{eq1thm} and \eqref{eq5thm} are trivially satisfied, 
since we have
\begin{subequations}
\begin{align}
\frac{\partial}{\partial u^k}\left(\eta^{is}\,\frac{\partial h^j}{\partial u^s}+\eta^{js}\,\frac{\partial h^i}{\partial u^s}\right)&=\eta^{is}\frac{\partial ^2 h^j}{\partial u^s\partial u^k}+\eta^{js}\frac{\partial^2 h^i}{\partial u^s\partial u^k}, \\[1ex]
        \frac{\partial b^{jr}_k}{\partial u^s}- \frac{\partial b^{jr}_s}{\partial u^k}&=\eta^{ja}\frac{\partial^3h^j}{\partial u^a\partial u^k \partial u^s}-\eta^{ja}\frac{\partial^3h^j}{\partial u^a\partial u^s \partial u^k}=0. 
    \end{align}
\end{subequations} 
We have to impose the further conditions \eqref{eq2thm} and \eqref{eq4thm} to prove the Hamiltonian property for $\mathcal{B}_{(1)}$. They give $G^{ijk}=0$ and $R^{jr}_{sk}=0$ respectively, with $G^{ijk}$ in \eqref{eq:cond_G} and $R^{jr}_{sk}$ in \eqref{r1}.

Now, we consider the conditions of Theorem~\ref{darv} in flat coordinates, i.e.\ conditions  \eqref{cond0a}, \eqref{cond1a} and~\eqref{cond3a}. E.g.\ the condition~\eqref{cond1a} becomes: 
\begin{align}\begin{split}
    P^{ijk}&=\eta^{is} \dfrac{\partial \omega_{\mathcal{B}}^{jk}}{ \partial u^s}+ \left(\eta^{il} \dfrac{\partial h^s}{\partial u^{l}}+\eta^{sl}\dfrac{\partial h^i}{\partial u^{l}} \right) \dfrac{\partial w_{\mathcal{A}}^{jk}}{\partial u^s}-\eta^{il} \dfrac{ \partial h^j}{\partial u^l \partial u^s}w_{\mathcal{A}}^{sk} -\eta^{il}  \dfrac{ \partial h^k}{\partial u^l \partial u^s}w_{\mathcal{A}}^{js}\\
    &~~+\eta^{js}  \dfrac{\partial w_{\mathcal{B}}^{ik}}{ \partial u^s} +\left(\eta^{jl}\dfrac{\partial h^s}{\partial u^{l}}+\eta^{sl}\dfrac{\partial h^j}{\partial u^{l}} \right)\dfrac{\partial w_{\mathcal{A}}^{ik}}{\partial u^s}-\eta^{jl}\dfrac{\partial h^i}{\partial u^{l} \partial u^s} w_{\mathcal{A}}^{sk} -\eta^{jl}\dfrac{\partial h^k}{\partial u^{l} \partial u^s} w_{\mathcal{A}}^{is}\,,
    \end{split}
\end{align}
and making $\omega_{\mathcal{A}}$ explicit as in~\eqref{eq:A_darboux}, i.e.\ $\omega^{ij}_\mathcal{A}=c^{ij}_su^s+f^{ij}$, we obtain \eqref{eq:cond_P}.
\end{proof}

\begin{remark}[On compatible first-order operators]
    We note that for operators in two components, the annihilation of tensors $G^{ijk}$ in \eqref{eq:cond_G} and $R^{jr}_{sk}$ in \eqref{r1} reduces to a system of hydrodynamic type in Liouville variables
    \begin{equation}\label{eq:liouv_var} 
        r^{ij}(u)=\eta^{ii}\,\dfrac{\partial h^j}{\partial u^i}\,,
    \end{equation}
    where the metric $\eta^{ij}$ is taken in its diagonal form.  This was shown and studied by Mokhov in~\cite{omokh}. For a compatible pair of homogeneous operators of first order $(\mathcal{A}_{(1)},\mathcal{B}_{(1)})$, the conditions for $\mathcal{B}_{(1)}$ in~\eqref{opmok} to be Hamiltonian are given by the vanishing of tensors $G^{ijk}$ in \eqref{eq:cond_G} and~$R^{jr}_{sk}$ in \eqref{r1} only, reproducing the result in~\cite[eq.\ (4.3),(4.4)]{Mok11}.
\end{remark}

 \begin{remark}[General coordinates expressions]\label{rmk:general_coords}
     It is easy to see that tensors $P^{ijk}$ in \eqref{eq:cond_P} and $S^{jk}_{ir}$ in \eqref{eq:cond_S} can be  written in general coordinates (non necessarily flat) as
    \begin{subequations} 
    \begin{align}
        P^{ijk}&=\nabla^i_\mathcal{A}\, \omega^{jk}_\mathcal{B}+\nabla_\mathcal{A}^{j}\, \omega^{ik}_\mathcal{B}+\nabla^i_\mathcal{B}\, \omega^{jk}_\mathcal{A}+\nabla_\mathcal{B}^{j}\, \omega^{ik}_\mathcal{A}\,, \label{newP}  \\[1ex]
        S^{jk}_{ir}&=\left(\nabla_{\mathcal{A}}\right)_{i}\left(\nabla_{\mathcal{A}}\right)_{r}\, \omega^{jk}_\mathcal{B}+\left(\nabla_{\mathcal{B}}\right)_{i}\left(\nabla_{\mathcal{B}}\right)_{r}\, \omega^{jk}_\mathcal{A}\,. \label{eq:S_geometrica}
    \end{align}
    \end{subequations} 
    \end{remark}

We consider now Corollary~\ref{cormok} for operators in $n=2$ and $n=3$ components and analyse the structure of $\mathcal{B}$.

\subsection{\texorpdfstring{Operator $\mathcal{B}$ in $2$ components}{n2}}
We consider the operator in the field variables $(u,v)\equiv(u^1,u^2)$. The non-degeneracy condition on the operator $\mathcal{A}$ as in~\eqref{eq:A_darboux} implies that the $(1+0)$ type operator is of the form 
\begin{equation}  \label{eq:A_nondeg_2comp}
    \mathcal{A}=\begin{pmatrix}
        a&0\\0&b
    \end{pmatrix}\partial_x+\begin{pmatrix}
        0&c\\-c&0
    \end{pmatrix},
\end{equation}
with $a,b,c \in \mathbb{R}$ constants. Here, we note that the only 2-dimensional Lie algebra compatible with a non-degenerate scalar product is the abelian one ($c^{ij}_k=0$), so that we can choose any arbitrary bi-vector $f^{ij}=c \, \partial_u \wedge \partial_v$.  In Darboux coordinates we only have four options $$(a,b)\in\{(1,1),(-1,1),(1,-1),(-1,-1)\},$$ 
from which we can distinguish two cases, i.e.\ $ab>0$ and $ab<0$. 

The form of the operator $\mathcal{B}$ is given by 
\begin{equation}
    \mathcal{B} = \mathcal{B}_{(1)} + \begin{pmatrix}
        0 & \omega(u,v) \\[1ex]
        -\omega(u,v) & 0
    \end{pmatrix}\,,
\end{equation}
with $\mathcal{B}_{(1)}$ dependent on the field variables locally via the functions $h^1(u,v)$, $h^2(u,v)$ as in~\eqref{opmok} for $2$ components, and $\omega_{\mathcal{B}}^{12}=-\omega_{\mathcal{B}}^{21}\equiv \omega(u,v)$. 

The following result holds true:
\begin{lemma}\label{lemma12}
    The ultralocal term $\omega_\mathcal{B}$ has its only non-zero entry given by
    \begin{equation}
        \omega(u,v)=c\left(\frac{\partial h^1}{\partial u}+\frac{\partial h^2}{\partial v}\right)+c_1\,,
            \label{formadiomega}
    \end{equation}
    with $c,c_1 \in \mathbb{R}$ constants. 
\end{lemma}
\begin{proof}
Exploiting Corollary \ref{cormok}, we consider the necessary condition $P^{112}=0$, i.e.
\begin{equation}
    \begin{split}\frac{1}{2}\,P^{112}&=\eta^{1s} \dfrac{\partial \omega^{12}_{\mathcal{B}}}{\partial u^s}+\left(\eta^{1\ell}\dfrac{\partial h^s}{\partial u^{\ell}}+\eta^{s\ell}\dfrac{\partial h^1}{\partial u^{\ell}}\right)\!\frac{\partial\omega_{\mathcal{A}}^{12}}{\partial u^{s}}-\eta^{1\ell} 
 \dfrac{\partial h^1}{\partial u^{\ell} \partial u^s}\omega^{s2}_{\mathcal{A}}-\eta^{1\ell}\dfrac{\partial h^2}{\partial u^{\ell} \partial u^s}\omega^{1s}_{\mathcal{A}}\\[.5ex]
   &=a \dfrac{\omega^{12}_{\mathcal{B}}}{\partial u}-a \dfrac{\partial^2 h^1}{\partial u^2}c-a 
 \dfrac{\partial^2 h^2}{\partial u \,\partial v}c = 0\,,
 \end{split}
\end{equation}
or equivalently
\begin{equation}
    \frac{\partial \omega}{\partial u}=c\left(\frac{\partial^2 h^1}{\partial u^2}+\frac{\partial h^{2}}{\partial u\,\partial v}\right)\,.
\end{equation}
Integrating with respect to $u$ we obtain
\begin{equation}\label{eq:omega1}
    \omega=c\left(\frac{\partial h^1}{\partial u}+\frac{\partial h^{2}}{\partial v}\right)+f_1(v),
\end{equation}
with $f_1(v)$ is arbitrary function in $v$. Analogously, for $P^{221}=0$ we have 
\begin{equation}\label{eq:omega2}
    \omega=c\left(\frac{\partial h^1}{\partial u}+\frac{\partial h^{2}}{\partial v}\right)+f_2(u),
\end{equation}
with $f_2(u)$ arbitrary function in $u$. For equations \eqref{eq:omega1} and \eqref{eq:omega2} to be both satisfied, we set $f_1(v)=f_2(u)=c_1$, with $c_1 \in \mathbb{R}$ constant, to give \eqref{formadiomega}. 
\end{proof}

We obtain the following classification of pairs $(\mathcal{A},\mathcal{B})$ of $(1+0)$ operators in $2$ components by making use of Mokhov's result (Theorem \ref{thm_mok}), 
the Corollary \ref{cormok} and the Lemma \ref{lemma12}.

\begin{theorem}\label{2x2thm}
    Any bi-Hamiltonian pair $(\mathcal{A},\mathcal{B})$ of non-homogeneous hydrodynamic operators with $\mathcal{A}$ non-degenerate, can be mapped into either $(\mathcal{A},\mathcal{B}^{\,1})$ or $(\mathcal{A},\mathcal{B}^{\,2})$, where $\mathcal{A}$ is of the form~\eqref{eq:A_nondeg_2comp} and 
    {\small
        \begin{equation} \label{eq:operator_B1}
            \mathcal{B}^{\,1}= \begin{pmatrix}
        2ak_1&bk_1+ak_2\\bk_1+ak_2&2bk_2
    \end{pmatrix}\partial_x
+\begin{pmatrix}
        0&c(k_1+k_2)+k_3\\
        -c(k_1+k_2)-k_3&0
    \end{pmatrix},
        \end{equation} }
        where $k_1,k_2,k_3 \in \mathbb{R}$ are arbitrary constants; 
        {\small
        \begin{equation}  \label{eq:operator_B2}
        \begin{split} 
            \mathcal{B}^{\,2}=&
            \begin{pmatrix}
                2a\,\dfrac{\partial h^1}{\partial u}&b\,\dfrac{\partial h^1}{\partial v}+a\,\dfrac{\partial h^2}{\partial u}\\[2ex]
                b\,\dfrac{\partial h^1}{\partial v}+a\,\dfrac{\partial h^2}{\partial u}&2b\,\dfrac{\partial h^2}{\partial v}
            \end{pmatrix}\partial_x
            +\begin{pmatrix}
                a\,\dfrac{\partial^2 h^1}{\partial u^2}&a\,\dfrac{\partial^2 h^2}{\partial u^2}\vspace{2mm}\\[2ex]
                b\,\dfrac{\partial^2 h^1}{\partial u\,\partial v}&b\,\dfrac{\partial^2 h^2}{\partial u\,\partial v}
            \end{pmatrix}u_x\\[1ex]
        &~+\begin{pmatrix}
                a\,\dfrac{\partial^2 h^1}{\partial u\,\partial v}&a\,\dfrac{\partial^2 h^2}{\partial u\,\partial v}\vspace{2mm}\\[1ex]
                b\,\dfrac{\partial^2 h^1}{\partial v^2}&b\,\dfrac{\partial^2 h^2}{\partial v^2}
            \end{pmatrix}v_x +\begin{pmatrix}
                0&c\left(\dfrac{\partial h^1}{\partial u}+\dfrac{\partial h^2}{\partial v}\right)\\[1ex]
                -c\left(\dfrac{\partial h^1}{\partial u}+\dfrac{\partial h^2}{\partial v}\right)&0
            \end{pmatrix},
            \end{split} 
        \end{equation}}
        where 
        \begin{enumerate}
            \item[\textup{1.}] if $ab>0$ in \eqref{eq:A_nondeg_2comp} then $h^1(u,v),h^2(u,v)$ are both solutions to the Laplace equation $\Delta h=0$
            
            \vspace*{-6ex}
            
            \begin{subequations}\label{eq:h1h2_laplace}
            \begin{align}
                h^{1}(u,v)&= \xi_1(u+ i v) + \xi_2(u-i v), \\ h^{2}(u,v)&= -i \xi_1(u+ i v) + i \xi_2\left(u-i v\right),
            \end{align}
            \end{subequations}
            with either $\xi_1^{''}=0$ or $\xi_2^{''}=0$;
            
            \item[\textup{2.}] if $ab<0$ then $h^1(u,v),h^2(u,v)$:
                       \begin{itemize}
                \item [{i.}]
            are solutions to the wave equation $\Box h=0$ with unitary velocity 
            \vspace*{-4ex}
            
            \begin{subequations}\label{eq:sol_waves}
            \begin{align}
                h^{1}(u,v)&=\xi_1(u+v)+\xi_2(u-v), \\ 
                h^{2}(u,v)&=\xi_1(u+v)-\xi_2(u-v)\,,
            \end{align}
            \end{subequations}
            with either $\xi_1''=0$ or $\xi_2''=0$, 
            \item[{ii.}] take the form
            
            \vspace*{-6ex}
            
            \begin{subequations}\label{eq:case2ii}
            \begin{align}
                h^1(u,v)&=\xi_1(u+v)+(v-u)\,\xi_2(u+v), \\ h^2(u,v)&=\xi_3(u+v)+(v-u)\,\xi_2(u+v)\,,
            \end{align}
            \end{subequations}
             \item[{iii.}]  take the form
             
            \vspace*{-6ex}
            
             \begin{subequations} \label{eq:case2iii}
            \begin{align}
               h^1(u,v)&=\xi_1(v-u)+(u+v)\, \xi_2(v-u), \\
               h^2(u,v)&=\xi_3(v-u)+(u+v)\, \xi_2(v-u)\,.
            \end{align} 
            \end{subequations} 
        \end{itemize}  
        \end{enumerate}
 
  \end{theorem}
  \begin{proof}
    We refer to Appendix \ref{class2x2}.  
  \end{proof}

  \begin{remark}[Compatibility and non-homogeneity]\label{rmk:compatibility_non_homog}
  We emphasise that the proposed classification is feasible thanks to the non-homogeneous nature of the operators in the compatibility expressed in terms of the Schouten bracket, i.e.
  \begin{equation}\label{eq:compatibility}
      [[\mathcal{A}+\lambda \,\mathcal{B},\mathcal{A}+\lambda \,\mathcal{B}]] = [[\mathcal{A},\mathcal{A}]] + \lambda ([[\mathcal{A},\mathcal{B}]] + [[\mathcal{B},\mathcal{A}]] ) + \lambda^2[[\mathcal{B},\mathcal{B}]]=0\,.
  \end{equation}
  The presence of the order $0$ terms produces a non-zero term linear in $\lambda$ in \eqref{eq:compatibility}, and hence further constraints (compare with \cite{DubFla,DubZha,Fer11,Mok12}). This technique was  used in a similar fashion in \cite{PavVerVit1}. 
  \end{remark}

Notice that $\mathcal{B}^{\,2}$ with $h^1(u,v),h^2(u,v)$ as in~\eqref{eq:h1h2_laplace} can be mapped into $\mathcal{B}^{\,2}$ with $h^1(u,v)$, $h^2(u,v)$ as in~\eqref{eq:sol_waves} with $v \mapsto iv$. The same happens for the cases \eqref{eq:case2ii}  and \eqref{eq:case2iii} with $v \mapsto -v$ (in this case the signature of the metric is unchanged). 
Therefore, all the admissible cases reduce either to the linear case ($\mathcal{B}^{\,1}$ in \eqref{eq:operator_B1}) or to the case where $h^1(u,v),h^2(u,v)$ are harmonic functions ($\mathcal{B}^{\,2}$ in \eqref{eq:operator_B2} with~\eqref{eq:h1h2_laplace}) when complex changes of variables are allowed.

As a by-product of the result, setting $c=0$ in~\eqref{eq:A_nondeg_2comp} and $c=0=k_3$ in~\eqref{eq:operator_B1} and~\eqref{eq:operator_B2}, the non-homogeneous pairs $(\mathcal{A},\mathcal{B}^{\,1})$ and $(\mathcal{A},\mathcal{B}^{\,2})$ reduce to the first-order pairs $(\mathcal{A}_{(1)},\mathcal{B}^{\,1}_{(1)})$ and $(\mathcal{A}_{(1)},\mathcal{B}^{\,2}_{(1)})$, for which we obtain a family of cases with no a priori assumptions on the degeneracy property of the metric~$g_{\mathcal{B}}$.

\begin{remark}[On Liouville variables]
    We observe that we recover the results obtained by Mokhov in~\cite{omokh} in terms of the Liouville variables $r^{ij}$~\eqref{eq:liouv_var} with some restrictions due to the presence of the additional structure of order $0$. In particular, the linear form of $h^1(u,v)$ and~$h^2(u,v)$ to get \eqref{eq:operator_B1} corresponds to~\cite[Theorem 2, case 1]{omokh} with the restriction that~$r^{11}$ and~$r^{22}$ are constants. The solutions~\eqref{eq:sol_waves} correspond to the solutions of the same case when~$ab < 0$. Finally, the solutions \eqref{eq:case2ii} and \eqref{eq:case2iii} satisfy the system of constraints introduced by Mokhov in \cite[Theorem 2, case 2]{omokh}. 
\end{remark}

\begin{remark}[Harmonic functions and Mokhov's result]
    In \cite{Mok13}, the author proved that a pair of metrics
    \begin{equation}
        g^{ij}_\mathcal{A}=\delta^{ij}, \qquad g^{ij}_{\mathcal{B}}=\text{e}^{a(u,v)}\,\delta^{ij},
    \end{equation}
    with $a(u,v)$ a non-constant real harmonic function are compatible if and only if $a=a(u\pm i\, v)$. We cover this result from the solution \eqref{eq:h1h2_laplace} of the Theorem \ref{2x2thm}. In particular, defining
    $$\xi_{1}'(u+iv)=\text{e}^{a(u,v)}+k_1, \qquad \text{or}\qquad \xi_2'(u-iv)=\text{e}^{a(u,v)}+k_2,$$
    with $k_1,k_2$ constants, 
    we obtain that the harmonic function $a(u,v)$ is of the same type as in~\cite[Proposition 4.6]{Mok13}.
\end{remark}

We conclude this part with some comments on compatible pairs whose second operator, namely $\mathcal{B}$, is degenerate. 

 \begin{remark}[On degenerate metrics]
 With $h^1(u,v),h^2(u,v)$ as in \eqref{eq:h1h2_laplace} and 
 \eqref{eq:sol_waves}, the metric $g_{\mathcal{B}}$ in the operator $\mathcal{B}^{\,2}$ \eqref{eq:operator_B2} is diagonal, since
\begin{equation}
    g^{12}_\mathcal{B}=b\,\dfrac{\partial h^1}{\partial v}+a\,\dfrac{\partial h^2}{\partial u}=0\,. 
\end{equation}
Moreover, in these cases we have (see Appendix \ref{class2x2} equations \eqref{eq:h1h2_constr_case1} and \eqref{eq:h1h2_constr_case2})
\begin{equation}
    \dfrac{\partial h^2}{\partial v}=\dfrac{\partial h^1}{\partial u},
\end{equation}implying that $g_\mathcal{B}$ is a conformally pseudo-Euclidean metric, i.e.
\begin{equation}
    g^{ij}_\mathcal{B}=2\dfrac{\partial h^1}{\partial u}\, \varepsilon^i\delta^{ij}, 
\end{equation}
where $\varepsilon\in\{-1,+1\}$. From this, it also follows that metrics $g_{\mathcal{B}}$ built from solutions \eqref{eq:h1h2_laplace} and~\eqref{eq:sol_waves} are non-degenerate if and only if either $h^1(u,v)=h^1(u)$ or $h^2(u,v)=h^2(v)$. This case, though, would fall into the linear case, i.e. the constant $\mathcal{B}^{\,1}$ in \eqref{eq:operator_B1}.

 In case $2.ii$ (e.g.\ with $(a,b)=(1,-1)$) the the solution \eqref{eq:sol_waves} induces a degenerate metric $g_{\mathcal{B}}$ if $\det g_{\mathcal{B}} = 0$, and e.g.\ solving in $\xi_2(u+v)$ we have 
\begin{equation}
    \xi_2(u+v) = \frac{1}{2} \left(\xi'_1(u+v) + \xi'_3(u+v) \right)\,, 
\end{equation}
and the metric becomes of rank 1 since the vectors $(g_{\mathcal{B}}^{11},g_{\mathcal{B}}^{12})$ and $(g_{\mathcal{B}}^{21},g_{\mathcal{B}}^{22})$ are proportional.

 \end{remark}

\subsection{\texorpdfstring{Operator $\mathcal{B}$ in $3$ components}{n3}}
A classification of non-homogeneous pairs in higher dimensions is, in general, a hard task. Even in the case with $n=3$ components we are not able to produce a list of compatible structures because of long non-trivial computations. Nevertheless, we show some preliminary results in this direction, referring to a future work on this subject.

We consider a pair~$(\mathcal{A},\mathcal{B})$ of non-homogeneous Hamiltonian operators in the variables $(u,v,w) \equiv (u^1,u^2,u^3)$, with $\mathcal{A}$ non-degenerate and written in Darboux form as
{\small
\begin{equation}  
       \mathcal{A}= \begin{pmatrix}
           a & 0 &0\\[1.5ex]           
           0 & b &0\\[1.5ex]
           0& 0 &c
       \end{pmatrix}\partial_x+
   \begin{pmatrix}
           0 & c_1 w + c_2& -\dfrac{c c_1}{b}v + c_3 \vspace{0.2cm}\\
           -c_1w - c_2&0 & \dfrac{ c c_1}{a} u + c_4 \vspace{0.2cm}\\
            \dfrac{c c_1}{b}v - c_3 & - \dfrac{ c c_1}{a} u - c_4 &0
                  \end{pmatrix},
   \end{equation} }

\noindent
   with $a,b,c\in\{-1,1\}$ and $c_1, \dots, c_4 \in \mathbb{R}$ arbitrary constants. We refer to \cite{GubOliSgrVer} for further details on the $n=3$ Lie algebra structures of the operators.

   Applying Mokhov's Theorem \ref{thm_mok} we obtain the form of the compatible operator of first-order $\mathcal{B}_{(1)}$, so that the metric $g_\mathcal{B}$ has the following entries:
   {\small
\begin{equation} \label{eq:metr_3comp}
g_\mathcal{B}=\begin{pmatrix}2a\,\dfrac{\partial h^1}{\partial u}&b\,\dfrac{\partial h^1}{\partial v}+a\,\dfrac{\partial h^2}{\partial u}&c\,\dfrac{\partial h^1}{\partial w}+a\,\dfrac{\partial h^3}{\partial u}\\[2ex]
b\,\dfrac{\partial h^1}{\partial v}+a\,\dfrac{\partial h^2}{\partial u}&2b\,\dfrac{\partial h^2}{\partial v}&c\,\dfrac{\partial h^2}{\partial w}+b\,\dfrac{\partial h^3}{\partial v}\\[2ex]
    c\,\dfrac{\partial h^1}{\partial w}+a\,\dfrac{\partial h^3}{\partial u}&c\,\dfrac{\partial h^2}{\partial w}+b\,\dfrac{\partial h^3}{\partial v}&2c\,\dfrac{\partial h^3}{\partial w}
    \end{pmatrix}\,,
\end{equation}} 
while the ultralocal term $\omega_{\mathcal{B}}$ has the form
{\small
\begin{equation}\label{eq:omega_3comp}
    \omega_\mathcal{B}=\begin{pmatrix}
        0&\omega_1&\omega_2\\
        -\omega_1&0&\omega_3\\
        -\omega_2&\omega_3&0
    \end{pmatrix}.
\end{equation}}
\begin{lemma}\label{lemma:ultralocal_3}
    The ultralocal term has the following entries
    {\small
    \begin{subequations}\label{eq:omega_3_entries}
    \begin{align}
        \begin{split} &\omega_1=\!\left(c_3 - \dfrac{c c_1}{b} v\right)\! \dfrac{\partial h^2}{\partial w} - \!\left(\dfrac{c c_1}{a} u + c_4\right)\! \dfrac{\partial h^1}{\partial w}- c_1 h^3
        +(c_2 + c_1 w) \!\left(\dfrac{\partial h^2}{\partial v}+\dfrac{\partial h^1}{\partial u} \right)\!+k w + k_1 \,,\end{split}\\[1ex]
       \begin{split}& \omega_2=\!\left(c_3 - \dfrac{c c_1}{b} v \right)\!\!\left(\dfrac{\partial h^1}{\partial u}+\dfrac{\partial h^3}{\partial w}  \right)\!+(c_2 + c_1 w)\dfrac{\partial h^3}{\partial v} 
       +\!\left(\dfrac{c c_1}{a} u + c_4\right)\! \dfrac{\partial h^1}{\partial v} +\dfrac{c c_1}{b}h^2-\dfrac{c k}{b}v + k_2\,,\end{split} \\[1ex]
       \begin{split} & \omega_3=\!\left(c_3 - \dfrac{c c_1}{b} v \right)\!\dfrac{\partial h^{2}}{\partial u}-(c_2 + c_1 w) \dfrac{\partial h^{3}}{\partial u} -\dfrac{c c_1}{a}h^1-\dfrac{c k}{a}u 
        + \!\left(\dfrac{c c_1}{a} u + c_4\right)\! \!\left(\dfrac{\partial h^{2}}{\partial v}+\dfrac{\partial h^{3}}{\partial w}\right)\!+ k_3\,. \end{split}
    \end{align}
    \end{subequations}}
\end{lemma}
\begin{proof}
We consider the vanishing tensor $P^{ijk}$ in \eqref{eq:cond_P}, in the non-trivial terms $P^{112}$ and $P^{122}$. 
    From 
    {\small
    \begin{equation}
    \begin{split}
    \frac{P^{112}}{a}&=\left(c_2+c_1w\right)\left(\frac{\partial^2h^2}{\partial u\partial v}+\frac{\partial^2 h^1}{\partial u^2}\right)-c_1\frac{\partial h^3}{\partial u}-\left(c_4+\frac{cc_1}{a}u\right)\frac{\partial^2h^1}{\partial u\partial w}-\frac{cc_1}{a}\frac{\partial h^1}{\partial w}\\
    &~~
    +\left(c_3-\frac{cc_1}{b}v\right)\frac{\partial^2h^2}{\partial u\partial w}-\frac{\partial \omega_1}{\partial u},
    \end{split}
    \end{equation}}
    the expression for the element $\omega_{1}(u,v,w)$ in \eqref{eq:omega_3comp}
    {\small
    \begin{equation}
        \omega_1=\left(c_3 - \dfrac{c c_1}{b} v\right) \dfrac{\partial h^2}{\partial w} - \left(\dfrac{c c_1}{a} u + c_4\right) \dfrac{\partial h^1}{\partial w}- c_1 h_3+ (c_2 + c_1 w) \left(\dfrac{\partial h^2}{\partial v}+\dfrac{\partial h^1}{\partial u} \right)+ F_1(v,w),
    \end{equation}}
    where $F_1$ is an arbitrary function in its arguments. Analogously, from
    {\small
    \begin{align}\begin{split}\frac{P^{122}}{b}&=\left(c_2+c_1w\right)\left(\frac{\partial^2h^2}{\partial v^2}+\frac{\partial^2 h^1}{\partial u \partial v}\right)-c_1\frac{\partial h^3}{\partial v}-\left(c_4+\frac{cc_1}{a}u\right)\frac{\partial^2h^1}{\partial v\partial w}-\frac{cc_1}{b}\frac{\partial h^2}{\partial w}\\
    &~~+\left(c_3-\frac{cc_1}{b}v\right)\frac{\partial^2h^2}{\partial v\partial w}-\frac{\partial \omega_1}{\partial v},\end{split}
    \end{align}}
    we obtain a similar expression, 
    {\small
    \begin{equation}
    \label{eq:omega_1_thm}
        \omega_1=\left(c_3 - \dfrac{c c_1}{b} v\right) \dfrac{\partial h^2}{\partial w} - \left(\dfrac{c c_1}{a} u + c_4\right)\dfrac{\partial h^1}{\partial w}- c_1 h_3+ (c_2 + c_1 w) \left(\dfrac{\partial h^2}{\partial v}+\dfrac{\partial h^1}{\partial u}\right)+ G_1(u,w),
    \end{equation}}
where $G_1$ is an arbitrary function in its arguments. Now, by comparison we have that 
$$F_1(v,w)=G_1(u,w) \quad \Longrightarrow \quad F_1=G_1=\theta_1(w).$$
Applying the same procedure to determine $\omega_2(u,v,w)$ and $\omega_3(u,v,w)$ in \eqref{eq:omega_3comp} (using  $P^{113}=0=P^{133}$ and $P^{223}=0=P^{233}$ respectively) we get 
{\small
\begin{align}
\label{eq:omega_2_thm}
     \omega_2&=\left(c_3 - \dfrac{c c_1}{b} v \right) \left(\dfrac{\partial h^1}{\partial u}+\dfrac{\partial h^3}{\partial w}\right)+(c_2 + c_1 w) \dfrac{\partial h^3}{\partial v} +\left(\dfrac{c c_1}{a} u + c_4\right) \dfrac{\partial h^1}{\partial v} +\dfrac{c c_1}{b}h_2+\theta_2(v)\,,\\
    \label{eq:omega_3_thm}
         \omega_3&=\left(c_3 - \dfrac{c c_1}{b} v \right)\dfrac{\partial h^2}{\partial u}-(c_2 + c_1 w) \dfrac{\partial h^3}{\partial u}+ \left(\dfrac{c c_1}{a} u + c_4\right) \left(\dfrac{\partial h^2}{\partial v}+\dfrac{\partial h^3}{\partial w}\right) -\dfrac{c c_1}{a}h_1+\theta_3(u)\,,
\end{align}}
where $\theta_2,\theta_3$ are arbitrary.

Finally, computing $P^{123}=0=P^{132}$ and using \eqref{eq:omega_1_thm}, \eqref{eq:omega_2_thm} and \eqref{eq:omega_3_thm} , we get
\begin{equation}
c\, \theta_1'(w)-a\, \theta_3'(u)=0, \qquad 
    b\, \theta_2'(v)+a\, \theta_3'(u)=0,
\end{equation}
or equivalently that $\theta_1,\theta_2,\theta_3$ are linear in their arguments,  where 
\begin{equation}
    \theta_1(w)=k\, w+k_1, \qquad 
    \theta_2(v)=-\frac{ck}{b}v+k_2, \qquad \theta_3(u)=\frac{ck}{a}u+k_3,
\end{equation}
with $k,k_1,k_2,k_3\in \mathbb{R}$ are arbitrary constants.

\end{proof}

We finally construct the KdV case as exposed in Example \ref{ex:inverted_kdv} by fixing some of the arbitrary constants and for suitable forms of the functions $h^k(u,v,w)$ with $k=1,2,3$.

\begin{example}[KdV equation]\label{ex:kdv_inverted}
Fixing  $a=1, \, b=-1,$  and $c=-1$ and choosing
\begin{subequations}
    \begin{align}
      &  h_1(u,v,w)= \dfrac{1}{4} u + \left(m_1-\dfrac{1}{2}\right)w +m_2 v+ m_3,\\[.5ex]
          &  h_2(u,v,w)= m_2 u-m_4  w + m_5,\\[1ex]
           &  h_3(u,v,w)=-\dfrac{1}{4}w + m_1u + m_4 v + m_6,
    \end{align}
    \end{subequations}
    where $m_1, \dots, m_6$ satisfy the relations 
    \begin{align*}
    &m_3 =\dfrac{1}{2}(k_3- \dfrac{1}{4}c_4  - c_2 m_1 + c_3 m_2 )-\dfrac{\sqrt{2}}{4},\\
           &m_5=-\dfrac{1}{2}\left(k_2 + c_2 m_4 + c_4 m_2\right),\\
           &m_6=-\dfrac{1}{2} \left( \dfrac{c_2}{4} + k_1 + c_4 \left(\dfrac{1}{2} - m_1 \right) - c_3 m_4 \right)+\dfrac{\sqrt{2}}{4}.
    \end{align*}
    With the particular choice $c_1=-2$, $k=0$ in \eqref{eq:omega_3_entries}, we reproduce exactly the second Hamiltonian structure of the KdV \eqref{kdvopsb}.
\end{example}
\begin{remark}
    The classification results presented in this Section provide a systematic framework for identifying both known and novel integrable models of non-homogeneous hydrodynamic type. 
    From a physical perspective, every bi-Hamiltonian pair $(\mathcal{A}, \mathcal{B})$ in the classification is a candidate for generating an integrable hierarchy. For any given pair, the associated bi-Hamiltonian system is obtained by identifying two Hamiltonian densities $F$ and $H$ such that the following relation holds:
\begin{equation}
u^i_t = \mathcal{A}^{ij} \frac{\partial F}{\partial u^j} = \mathcal{B}^{ij} \frac{\partial H}{\partial u^j}.
\end{equation}
 A possible procedure to achieve this aim is to solve the compatibility conditions introduced by Tsarev \cite{tsarev91:_hamil} for homogeneous systems and extend them for non-homogeneous ones, as also studied in \cite{Ver3}. As an example, in the two-component case, this requirement leads to a system of four partial differential equations in the unknowns $F(u,v)$ and $H(u,v)$. Preliminary computations indicate the existence of non-trivial solutions for $F$ and $H,$ ensuring that the classification lists are directly connected to the discovery of new integrable $(1+0)$ systems. Nevertheless, achieving a general characterisation of these integrable system requires substantial additional work and therefore it will be object of a forthcoming paper.

\end{remark}

\section{\texorpdfstring{Geometric description of bi-Hamiltonian $(1+0)$ operators}{bihamnije}}\label{biham_nije}

In this final Section, we introduce the notion of the bi-pencil defined for a pair of $(1+0)$ operators, and we present some results for non-homogeneous operators in connection with Nijenhuis geometry. Both these analyses are in the direction of a purely geometric description of the $(1+0)$ operators and of the pairs of such operators. 

\subsection{Bi-pencils}\label{bipen}
Here, we consider both operators in the pair $(\mathcal{A},\mathcal{B})$ to be non-degenerate. Although we have not found any physically relevant example of a bi-Hamiltonian system with this property, this case 
has a meaningful geometric interpretation in terms of compatible pairs of non-homogeneous operators. 

First, we observe that if $\mathcal{A}$ and $\mathcal{B}$ are non-degenerate and form a compatible pair, then their first-order restrictions $(\mathcal{A}_{(1)},\mathcal{B}_{(1)})$ are a compatible pair of Dubrovin-Novikov operators, whereas the ultralocal terms $(\omega_\mathcal{A},\,\omega_\mathcal{B})$ as zero-order restriction $(\mathcal{A}_{(0)},\mathcal{B}_{(0)})$ must form a compatible pencil of Poisson structures. The definition of Poisson pencil was introduced by Magri and Morosi in \cite{MaMo}, where the authors found a characterization of pairs $(\omega_\mathcal{A},\,\omega_\mathcal{B})$ such that any linear combination $\omega_\mathcal{A}+\lambda\, \omega_\mathcal{B}$ is still a Poisson tensor, i.e.\ it satisfies the Jacobi identity\footnote{The skew-symmetry is automatically ensured by the skew-symmetric property of $\omega_\mathcal{A}$ and $\omega_\mathcal{B}$.}. 
A similar object was introduced by Dubrovin in \cite{dub2}, who extended the notion of pencil for contravariant metrics. In particular, defining 
    \begin{equation}
    \Gamma^{ij}_k = - g^{is}\,\Gamma^{j}_{sk}\,, \qquad
    R^{ij}_{k\ell} = g^{is} R^j_{sk\ell}\,,
\end{equation}
where $\Gamma^{j}_{sk}$ and $R^j_{sk\ell}$ are the Christoffel symbols and the curvature for the metric $g_{ij}$ respectively, 
and introducing 
\begin{equation}\label{eq:g_mu}
g^{ij}_\mu:=g^{ij}_\mathcal{A}+\mu \,g^{ij}_{\mathcal{B}},   
\end{equation}
in \cite{dub2}, the author showed that two Dubrovin-Novikov operators are compatible if and only if $g_\mathcal{A}$ and $g_\mathcal{B}$ are compatible as metrics, i.e.\ if the two following conditions are satisfied
\begin{align}
\Gamma^{ij}_{\mu,k}&=\Gamma^{ij}_{\mathcal{A},k}+\mu\,\Gamma^{ij}_{\mathcal{B},k},\qquad 
    R^{ij}_{\mu, kl}=R^{ij}_{\mathcal{A},kl}+\mu\, R^{ij}_{\mathcal{B},kl}.
\end{align}

In this Section, we give the notion of compatible bi-pencils for non-homogeneous operators of type $(1+0)$. With this aim, we introduce 
\begin{equation}\label{eq:omega_mu}
        \omega_\mu^{ij}:=\omega^{ij}_\mathcal{A}+\mu\,  \omega^{ij}_\mathcal{B},
    \end{equation}
and we recall that a Killing-Yano $(2,0)$-tensor for the metric $g$ is a skew-symmetric tensor $\pi^{ij}$ satisfying
\begin{equation}\label{eq:KY_cond}
\nabla^i\,\pi^{jk}+\nabla^j\,\pi^{ik}=0. 
\end{equation}
\begin{definition}
    Given two contravariant metrics $g_\mathcal{A}$, $g_\mathcal{B}$ and two ultralocal structures $\omega_\mathcal{A}$ and $\omega_\mathcal{B}$, we say that they form the bi-pencil $(g_\mu,\omega_\mu)$ if the metrics form a pencil $g_\mu$ as in~\eqref{eq:g_mu}, and the ultralocal terms form a pencil $\omega_\mu$ as in~\eqref{eq:omega_mu}, which is a Killing-Yano tensor for the metric $g_\mu$. 
\end{definition}
    
    \begin{remark}
        By definition, it follows that not any pair of pencils $g_\mu$, $\omega_\mu$ forms a bi-pencil. The additional geometric requirement for $\omega_\mu$ to be a Killing-Yano tensor is, indeed, nontrivial.
    \end{remark}
    Moreover, if both operators in the pair $(\mathcal{A},\mathcal{B})$ are non-degenerate 
    the following result holds. 
\begin{theorem}\label{thm:bi_pencil}
   Assuming $\det(g_\mathcal{A}+\mu\, g_\mathcal{B})\neq 0$, the compatibility of two non-degenerate Hamiltonian operators $\mathcal{A}$ and $\mathcal{B}$ is equivalent to requiring that their metrics $g_{\mathcal{A}}$, $g_{\mathcal{B}}$ and the ultralocal terms $\omega_{\mathcal{A}}$, $\omega_{\mathcal{B}}$ form the bi-pencil $(g_{\mu},\omega_{\mu})$.
\end{theorem}
\begin{proof}
    Using the results of Magri and Morosi and Mokhov, the first-order terms $(\mathcal{A}_{(1)},\mathcal{B}_{(1)})$ and the zero-order ones $(\mathcal{A}_{(0)},\mathcal{B}_{(0)})$ are compatible if and only if $(g_{\mathcal{A}},g_{\mathcal{B}})$ and $(\omega_{\mathcal{A}},\omega_{\mathcal{B}})$ form two pencils, i.e.\ $g_{\mu}$ and $\omega_{\mu}$ respectively. 

    We define $\nabla_{\!\mu}$ to be the covariant derivative  with respect to $g_\mu$. We now observe that for any parameter $\mu$ we can write the following: 
    \begin{align*}
        \begin{split}
\nabla^i_\mu\omega^{jk}_\mu+\nabla^j_\mu\omega^{ik}_\mu&
        =g^{is}_\mu \dfrac{\partial \omega^{jk}}{\partial u^s}+g^{js}_\mu\dfrac{\partial \omega^{ik}}{\partial u^s}-\Gamma^{ij}_{\mu,s}\omega^{sk}-\Gamma^{ik}_{\mu,s}\omega^{js}-\Gamma^{ji}_{\mu,s}\omega^{sk}-\Gamma^{jk}_{\mu,s}\omega^{is}\\
        &=\left(g^{is}_\mathcal{A}+\mu\, g^{is}_\mathcal{B}\right)\left(\dfrac{\partial \omega_{\mathcal{A} }^{jk}}{\partial u^s}+\mu\,  \dfrac{\partial \omega_{\mathcal{B} }^{jk}}{\partial u^s}\right)+\left(g^{js}_\mathcal{A}+\mu\, g^{js}_\mathcal{B}\right)\left(\dfrac{\partial \omega_{\mathcal{A} }^{ik}}{\partial u^s}+\mu\,  \dfrac{\partial \omega_{\mathcal{B} }^{ik}}{\partial u^s}\right)\\
        &\hphantom{ciao}-\left(\Gamma^{ji}_{1,s}+\mu\, \Gamma^{ji}_{2,s}\right)\left(\omega^{sk}_{\mathcal{A}}+\mu\, \omega^{sk}_{\mathcal{B}}\right)-\left(\Gamma^{jk}_{1,s}+\mu\, \Gamma^{jk}_{2,s}\right)\left(\omega^{is}_{\mathcal{A}}+\mu\, \omega^{is}_{\mathcal{B}}\right)\\
        &\hphantom{ciao}-\left(\Gamma^{ij}_{1,s}+\mu\, \Gamma^{ij}_{2,s}\right)\left(\omega^{sk}_{\mathcal{A}}+\mu\, \omega^{sk}_{\mathcal{B}}\right)-\left(\Gamma^{ik}_{1,s}+\mu\, \Gamma^{ik}_{2,s}\right)\left(\omega^{js}_{\mathcal{A}}+\mu\, \omega^{js}_{\mathcal{B}}\right)\\
        &=\left(\nabla^i_\mathcal{A}\, \omega^{jk}_\mathcal{A}+\nabla_\mathcal{A}^{j}\, \omega^{ik}_\mathcal{A}\right)+\mu\, P^{ijk}+ \mu^2\, \left(\nabla^i_\mathcal{B}\, \omega^{jk}_\mathcal{B}+\nabla_\mathcal{B}^{j}\, \omega^{ik}_\mathcal{B}\right),\end{split}
    \end{align*}
so that for Hamiltonian operators $\mathcal{A}$ and $\mathcal{B}$, the vanishing of tensor $P^{ijk}$ in~\eqref{eq:cond_P} is equivalent to the Killing-Yano condition~\eqref{eq:KY_cond} for $\omega_\mu$ in the bi-pencil. 

Finally, to satisfy all the conditions of Theorem \ref{darv}, it remains to verify the vanishing of tensor $S^{ijk}_{r}$ in~\eqref{eq:cond_S}. However, as remarked by Mokhov and Ferapontov in \cite{FerMok1}, this condition is a formal corollary of the Killing-Yano property for $g_{\mu}$. 
\end{proof}

\begin{remark}[On strong bi-pencils]\label{rmk:strong_bi_pencil}
    One can choose to use a stronger definition of bi-pencil in which every linear combination of compatible Poisson tensors $\omega_\lambda$ is a Killing-Yano tensor for every metric $g_\mu$ in the pencil, i.e. 
    \begin{equation}\label{eq:KY_our}
    \nabla_{\mu}^i\,\omega^{jk}_\lambda+\nabla_\mu^j\,\omega^{ik}_\lambda=0,
    \end{equation} we call such structure strong bi-pencil. 
   In this case, the Killing-Yano condition~\eqref{eq:KY_our} reads as
    \begin{align*}
        \begin{split}\nabla^i_\mu\,\omega^{jk}_\lambda+\nabla^j_\mu\,\omega^{ik}_\lambda=\left(\nabla^i_\mathcal{A}\, \omega^{jk}_\mathcal{A}+\nabla_\mathcal{A}^{j}\, \omega^{ik}_\mathcal{A}\right)+\mu\, P^{ijk}_1+\lambda\, P^{ijk}_2+ \mu\lambda\, \left(\nabla^i_\mathcal{B}\, \omega^{jk}_\mathcal{B}+\nabla_\mathcal{B}^{j}\, \omega^{ik}_\mathcal{B}\right),\end{split}
    \end{align*}
    where 
    \vspace*{-2ex}
    \begin{subequations}
    \begin{align}
        P^{ijk}_1&=\nabla^i_\mathcal{B}\, \omega^{jk}_\mathcal{A}+\nabla_\mathcal{B}^{j}\, \omega^{ik}_\mathcal{A}\,,\\
        P^{ijk}_2&=\nabla^i_\mathcal{A}\, \omega^{jk}_\mathcal{B}+\nabla_\mathcal{A}^{j}\, \omega^{ik}_\mathcal{B}\,,
    \end{align}
    \end{subequations}
    hence any strong bi-pencil defines a bi-Hamiltonian pair of non-homogeneous operators. The converse of this is in general not true. We indicate such a bi-pencil with $(g_\mu,\omega_\lambda)$ with different indices, to distinguish it from the previous one. Both cases are sketched in Figure~\ref{fig:bipencil}.
\end{remark}

\begin{figure}[t]
\centering 
\includegraphics[width=.4\textwidth]{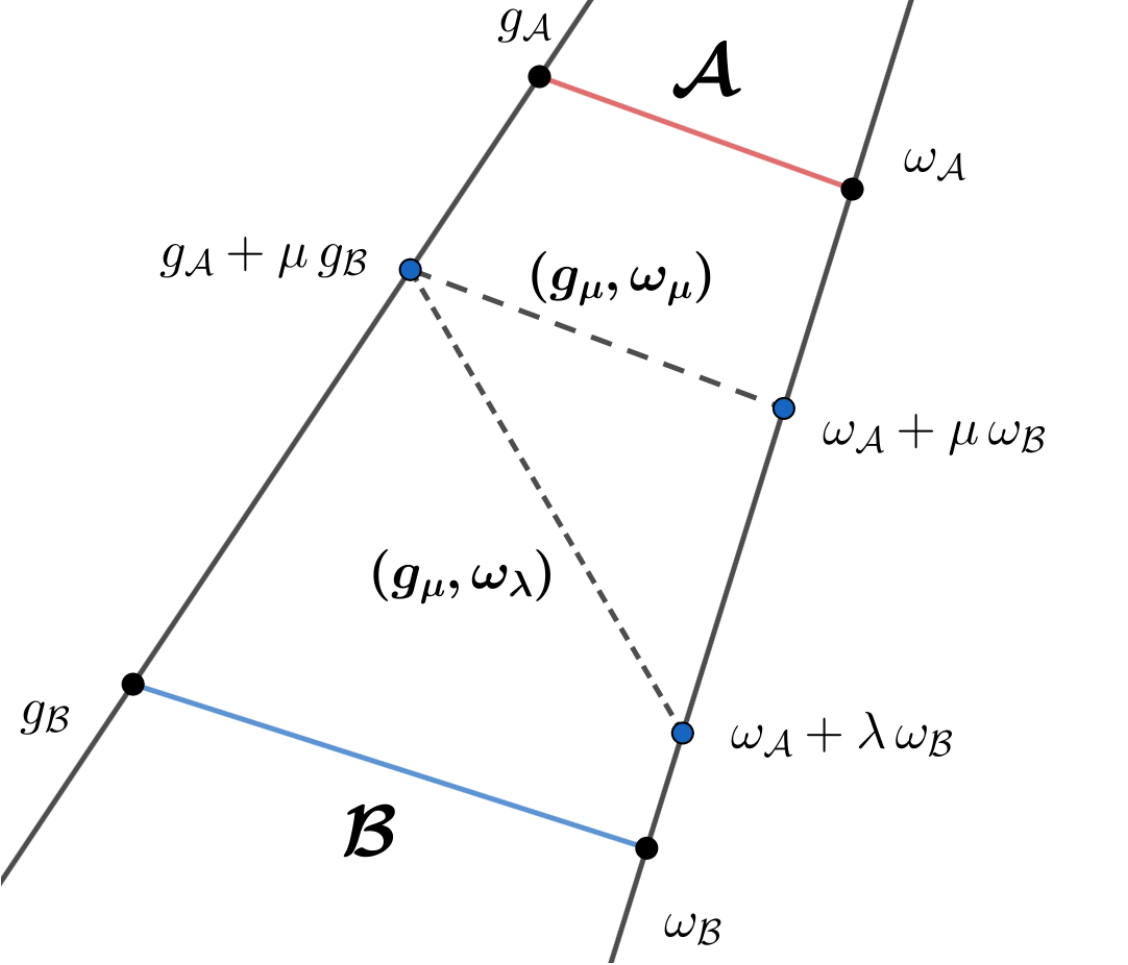}
\caption{Sketch of the bi-pencil $(g_{\mu},\omega_{\mu})$ and the strong bi-pencil $(g_{\mu},\omega_{\lambda})$.}
\label{fig:bipencil}
\end{figure}

\begin{example}
  We consider the case of strong bi-pencils in $n=2$ components. Adding the constraint $P_1^{ijk}=P_2^{ijk}=0$ to the conditions of Theorem~\ref{2x2thm}, the only compatible pair $(\mathcal{A},\mathcal{B})$ is given by constant operators. Explicitly, in $n=2$ components, strongly bi-pencils can be mapped into the following two pairs:
  \begin{equation}
    \mathcal{A}=  \begin{pmatrix}
          a&0\\0&b
      \end{pmatrix}\partial_x+\begin{pmatrix}
          0&c\\-c&0
      \end{pmatrix},\qquad 
      \mathcal{B}=\begin{pmatrix}
          k_1&k_2\\k_2&k_3
      \end{pmatrix}\partial_x+\begin{pmatrix}
          0&k_4\\-k_4&0
      \end{pmatrix}.
  \end{equation}
\end{example}
\begin{remark}
    For the pair $(\mathcal{A},\mathcal{B})$ in $n=3$ components, the condition $P^{ijk}_2=0$ implies that $\omega_{\mathcal{B}}(u,v,w)$ is linear, with 
    \begin{equation}
        \omega^{12}_{\mathcal{B}}=k\, w+k_1, \qquad \omega^{13}_{\mathcal{B}}=-\frac{ck}{b}v+k_2, \qquad \omega^{23}_{\mathcal{B}}=\frac{ck}{a}u+k_3,
    \end{equation}
    with $c, k, k_1, k_2, k_3$ constants, whereas $P^{ijk}_1=0$ implies more sophisticated conditions which are hard to integrate. In particular, the functions $h^1(u,v,w),h^2(u,v,w),h^3(u,v,w)$ appearing in $\mathcal{B}_{(1)}$ in \eqref{eq:metr_3comp} in this case are independent of $\omega_\mathcal{B}$, and this leads to the old problem of classifying first-order compatible pairs. 
    However, similarly to what happens for $n=2$ components, for $n=3$ components a class of strong bi-pencils is represented by the pair of constant operators
      \begin{equation*}
    \mathcal{A}=  \begin{pmatrix}
          a&0 &0\\0&b &0\\0&0 &c
      \end{pmatrix}\partial_x+\begin{pmatrix}
       0&c_1 &c_2\\-c_1&0 &c_3\\-c_2&-c_3 &0
      \end{pmatrix},\qquad 
      \mathcal{B}=\begin{pmatrix}
        k_1&k_2 &k_3\\k_2&k_4 &k_5\\k_3&k_5 &k_6
      \end{pmatrix}\partial_x+\begin{pmatrix}
         0&k_7 &k_8\\  -k_7&0 &k_9\\ -k_8&-k_9 &0
      \end{pmatrix},
  \end{equation*}  
  where $k_i$ ($i=1, \dots, 8$) and $c_j$ ($j=1,2,3$) are arbitrary constants. 
  \end{remark}

\subsection{Nijenhuis geometry} 
In this last Section, we present some preliminary results on non-homogeneous operators within the context of the Nijenhuis geometry. The connection between Nijenhuis tensors (i.e.\ $(1,1)$-tensor with zero Nijenhuis torsion) and integrable systems has been deeply investigated (see the recent papers ~\cite{Bolsi0,Bolsi1,Bolsi2,Bolsi3,Bolsi4,Bolsi5} and \cite{BolsiReview} for an introductory treatment of the topic). 

In \cite{MaMo}, Magri and Morosi proved that two Poisson bivectors $\omega_\mathcal{A}$ and $\omega_\mathcal{B}$ are compatible (i.e.\ they form a pencil) if and only if in the non-degenerate case the $(1,1)$-tensor
\begin{equation}
    L^i_j=\omega_\mathcal{A}^{is}\,(\omega_{\mathcal{B}})_{sj},
\end{equation}
where $(\omega_{\mathcal{B}})_{is}=(\omega_{\mathcal{B}}^{is})^{-1}$ has zero Nijenhuis torsion $\mathcal{N}_L=0$:
\begin{equation}\label{nij}
\mathcal{N}_L(v,w)= L^2[v, w] + [Lv, Lw] - L[Lv, w] - L[v, Lw]=0,
\end{equation}
for any vector fields $v,w$. In local coordinates, the expression \eqref{nij} can be written as 

    \begin{equation}
N^{k}_{ij}=L^s_i\frac{\partial L^{k}_{j}}{\partial u^s}-L^s_{j}\frac{\partial L^{k}_{i}}{\partial u^s}+L^k_s\frac{\partial L^s_{i}}{\partial u^j}-L^k_s\frac{\partial L^{s}_{j}}{\partial u^i}.
\end{equation}

For first-order operators a similar criterion was  introduced by Ferapontov in \cite{Fer11} and {further investigated} by Mokhov in \cite{Mok14}. We collect the main results in the following theorem:
\begin{theorem}[\cite{Mok14}]
    Let $g_\mathcal{A}$ and $g_\mathcal{B}$ be two non-degenerate metrics which are non-singular, i.e.\ for which the solutions $\lambda$ in 
    \begin{equation}
        \det(g_\mathcal{A}+\lambda \, g_\mathcal{B})=0
    \end{equation}
    are distinct. The metrics $g_\mathcal{A}$ and $g_\mathcal{B}$ form a pencil of contravariant metric if and only if the tensor 
    \begin{equation}\label{eq:Nij_tens}
    A^i_j=g_\mathcal{A}^{is}\,(g_{\mathcal{B}})_{sj}    
    \end{equation}
    is a Nijenhuis tensor.  
\end{theorem}
More generally, it holds true that two compatible metrics have vanishing Nijenhuis torsion (as pointed out by Mokhov, the converse is not in general true without the additional requirement to be non-singular).

As also expected, but in a strongly non-trivial way we can consider the following:
\begin{example}[Compatible metrics in the $n=2$ classification]

Let us consider the metrics from the restriction $\mathcal{A}_{(1)}$ of $\mathcal{A}$ in \eqref{eq:A_nondeg_2comp} and $\mathcal{B}^{\,2}_{(1)}$ of $\mathcal{B}^{\,2}$ in \eqref{eq:operator_B2}
\begin{equation}\label{eq:gA_gB2}
    g^{ij}_\mathcal{A}=\begin{pmatrix}
        a&0\\0&b
    \end{pmatrix}, \qquad g_{\mathcal{B}^{2}}^{ij}=\begin{pmatrix}
                2a\,\dfrac{\partial h^1}{\partial u}&b\,\dfrac{\partial h^1}{\partial v}+a\,\dfrac{\partial h^2}{\partial u}\\[2ex]
                b\,\dfrac{\partial h^1}{\partial v}+a\,\dfrac{\partial h^2}{\partial u}&2b\,\dfrac{\partial h^2}{\partial v}
            \end{pmatrix},
\end{equation}
forming a pencil of metrics. One can compute the Nijenhuis torsion for the (1,1) defined by $g_{\mathcal{A}}^{-1}$ and $g_{\mathcal{B}^2}$ as they are. In this case the torsion does not vanish identically, indeed only specifying the functions $h^1(u,v)$ and $h^2(u,v)$ as in Theorem \ref{2x2thm} one obtains the vanishing torsion. This is due to the fact that the Mokhov form~\eqref{opmok} for a compatible pair is only necessary and not sufficient to find bi-Hamiltonian first-order pairs.   
    
\end{example}

\begin{remark}[On singular metric pencils]
    We stress that the compatible metrics found in Theorem \ref{2x2thm} do not form, in general, non-singular pairs. In particular, setting to zero the determinant of $(g_{\mathcal{A}}+\lambda g_{\mathcal{B}^2})$ with~\eqref{eq:gA_gB2}, we obtain a polynomial of degree two in $\lambda$:
\begin{equation}
   \bigg(\!-\!\left(b \dfrac{\partial h^{1}}{\partial v}+ a \dfrac{\partial h^{2}}{\partial u} \right)^{\!\!2}+ 4 a b \, \dfrac{\partial h^{2}}{\partial v} \, \dfrac{\partial h^{1}}{\partial u} \bigg)\lambda^{2} +2 a  b \,\left(\dfrac{\partial h^{2}}{\partial v} + \dfrac{\partial h^{1}}{\partial u} \right) \lambda +a \, b=0,
\end{equation}
whose discriminant is
 \begin{equation}\label{eq:discriminant}
     \dfrac{\Delta}{4}=  a  b \,\left(b \dfrac{\partial h^{1}}{\partial v}+ a \dfrac{\partial h^{2}}{\partial u} \right)^{\!\!2} + \left(\dfrac{\partial h^{2}}{\partial v}-\dfrac{\partial h^{1}}{\partial u}\right)^{\!\!2}.
 \end{equation}
For pairs of type $(\mathcal{A},\mathcal{B}^{2})$ in Theorem \ref{2x2thm} the expression in~\eqref{eq:discriminant} vanishes identically. We then obtain compatible metrics which are singular. Whereas, for pairs of type $(\mathcal{A},\mathcal{B}^{1})$ in Theorem~\ref{2x2thm} (with the two metrics being both constant), we obtain 
\begin{equation}
     \dfrac{\Delta}{4}=a b \, (b k_2+a k_4)^2+ (k_1-k_5)^2,
\end{equation}
so that the results on $\lambda$ strictly depend on the choice of the arbitrary constants.

\end{remark}

It is generally believed that Nijenhuis geometry is strictly related to compatible pairs and integrable systems. However,  there is a lack of knowledge in this direction regarding non-homogeneous Hamiltonian structure. We plan to investigate this topic in a future paper. However, a preliminary result has been obtained for a single non-degenerate operator $\mathcal{A}$. We recall that in Darboux coordinates, $g_\mathcal{A}$ reduces to be constant and, as already mentioned in Remark \ref{darboux}, the operator takes the form~\eqref{eq:A_darboux}.

\begin{remark}[On left-symmetric Lie algebras]
   It is a remarkable fact that in \cite{koni1} the author proved that given a real algebra $\mathfrak{a}$ of dimension $n$ the $(1,1)$-tensor defined as
   \begin{equation}
       T^i_j=a^{i}_{js}x^s,
   \end{equation}
   (i.e.\ the right-adjoint operator of $\mathfrak{a}$) 
   has zero Nijenhuis torsion if and only if the algebra is  left-symmetric. 
\end{remark}

In our context, the following result holds true.

\begin{theorem}\label{nijnonhom}
    The affinor $L^i_j=g_{js}\left(c^{si}_\ell u^\ell+f^{si}\right)$ has vanishing Nijenhuis torsion \eqref{nij} if and only if
   \begin{subequations}\label{condf} 
   \begin{align}
        &c^{is}_jc^{k\ell}_s=0,\label{cond1x}\\[1ex]
        &f^{is}c_s^{jk}=0. \label{cond2x}
    \end{align}
    \end{subequations}
\end{theorem}
\begin{proof}
We refer to Appendix \ref{nijapp}.    
\end{proof}
 The local expressions \eqref{condf} have a purely geometric interpretation in the contest of Lie algebra structures:
\begin{corollary}
    The Nijenhuis torsion  of the operator $L^i_j$ vanishes if and only if the Lie algebra structure is 2-step nilpotent and the extension given by the the 2-cocycle is 2-step nilpotent too.
\end{corollary}

We now present an example:

\begin{example}Following the classification introduced in \cite{GubOliSgrVer}, we consider the first case of $2$-step nilpotent Lie algebra, so that the following operator satisfies the vanishing of its Nijenhuis torsion:
\begin{equation}
   \left(\begin{array}{cccccc}
0 & 0 & 0 & \alpha & 0 & 0 
\\
 0 & 0 & 0 & 0 & 0 & \alpha 
\\
 0 & 0 & 0 & 0 & -\alpha & 0 
\\
\alpha & 0 & 0 & \beta & \gamma & \delta 
\\
 0 & 0 & -\alpha & \gamma & \lambda & \epsilon 
\\
 0 & \alpha & 0 & \delta & \epsilon & \mu
\end{array}\right)\partial_x+\begin{pmatrix}
    0&0&0&0&0&0\\0&0&0&0&0&0\\0&0&0&0&0&0\\0&0&0&0&u^2&u^3\\
    0&0&0&-u^2&0&u^1\\
    0&0&0&-u^3&-u^1&0
\end{pmatrix},
\end{equation}
where $\alpha, \beta, \gamma,\delta, \epsilon$ and $\mu$ are arbitrary constants and the operator is defined on a $6$-dimensional manifold of field variables $u^1,\dots u^6$.
\end{example}

We stress that Theorem \ref{nijnonhom} does not represent a solution of the problem of understanding non-homogeneous operators in the contest of Nijenhuis geometry which, on the contrary, we believe should be a meaningful tool to deal with these operators. In this direction, we conjecture the existence of a non-trivial tensor $R^i_j$  whose vanishing Nijenhuis torsion is equivalent to the vanishing of the tensor $P^{ijk}$, thereby extending the result obtained for pencils of metrics and of ultralocal structures to the case of bi-pencils. This will be investigated in a future work.

\subsection*{Acknowledgements} 
The authors are extremely thankful to N.\ Manganaro, D.\ Akpan, E.\ Sgroi and R.\ Vitolo  for their comments, suggestions and for their concrete help.  We finally acknowledge the financial support of GNFM of the Istituto Nazionale di Alta Matematica. {AR is partially supported by the PRIN project MIUR Prin 2022, project code
	1074 2022M9BKBC, Grant No. CUP B53D23009350006}. PV is partially funded by the research project Mathematical Methods in Non-Linear Physics (MMNLP) by the Commissione Scientifica Nazionale – Gruppo 4 – Fisica Teorica of the Istituto Nazionale di Fisica Nucleare (INFN) Sezione di Lecce and by ``Borse per viaggi all'estero'' of the Istituto Nazionale di Alta Matematica, which permitted to visit the Geometry and Physics group of Loughborough University.

\vspace{5mm}

\paragraph{Data Availability} Not applicable. \paragraph{Declarations}

 The authors have no competing interests to declare that are relevant to the content of this article.

\paragraph{Ethical Statement}

No statement required.

\newpage 
\begin{appendices}

\section{\texorpdfstring{Classification of operators $C^{ij}_{\ell,k}$}{operCij}}\label{app:old_classification}
We list here the classification obtained in \cite{DellAVer1}. The operators $C^{ij}_{3,2}$ and $C^{ij}_{3,5}$ can be found in the text in \eqref{eq:operator_C32} and \eqref{eq:operator_C35}. 
{\footnotesize
\begin{equation*}
    \begin{array}{c c c c}
        \hline \\[-.5ex] 
        \text{Operator} & g^{ij} \, \partial_x & b^{ij}_k\,u^k_x  &  \omega^{ij} \\[1ex]
        \hline \\[-.5ex] 
        C^{ij}_{2,1} & 
        \begin{pmatrix}
        \partial_x & 0 \\
        0 & 0
        \end{pmatrix} & 
        \begin{pmatrix}
        0 & 0 \\
        0 & 0
        \end{pmatrix} & 
        \begin{pmatrix}
        0 & f(v) \\
        -f(v) & 0
        \end{pmatrix} \vspace{1ex} \\[1ex] \hline \\[-.5ex] 
        C^{ij}_{2,2} & 
        \begin{pmatrix}
        \partial_x & 0 \\
        0 & 0
        \end{pmatrix} &
        \begin{pmatrix}
        0 & -\dfrac{v_x}{u} \\
        \dfrac{v_x}{u} & 0
        \end{pmatrix} &
        \begin{pmatrix}
        0 & \dfrac{f(v)}{u} \\
        -\dfrac{f(v)}{u} & 0
        \end{pmatrix} 
        \vspace{2ex} \\[1ex] \hline \\[-.5ex] 
        C^{ij}_{3,1} 
        &
\begin{pmatrix}
0 & 0 & 0 \\
0 & 0 & 0 \\
0 & 0 & 0
\end{pmatrix} &
\begin{pmatrix}
0 & w_x & 0 \\
-w_x & 0 & 0 \\
0 & 0 & 0
\end{pmatrix} 
&
\begin{pmatrix}
0 & f(u,v,w) & 0 \\
-f(u,v,w) & 0 & 0 \\
0 & 0 & 0
\end{pmatrix}
 \vspace{2ex} \\[1ex] \hline \\[-.5ex] 
 C^{ij}_{3,3} & 
\begin{pmatrix}
\partial_x & 0 & 0 \\
0 & 0 & 0 \\
0 & 0 & 0
\end{pmatrix}
&
\begin{pmatrix}
0 & w_x & 0 \\
-w_x & 0 & 0 \\
0 & 0 & 0
\end{pmatrix}
&
\begin{pmatrix}
0 & f(v,w) & 0 \\
-f(v,w) & 0 & 0 \\
0 & 0 & 0
\end{pmatrix}
 \vspace{2ex} \\[1ex] \hline \\[-.5ex] 

C^{ij}_{3,4}  &
\begin{pmatrix}
\partial_x & 0 & 0 \\
0 & 0 & 0 \\
0 & 0 & 0
\end{pmatrix}
&
\begin{pmatrix}
0 & 0 & -\dfrac{w_x}{u} \\
0 & 0 & 0 \\
\dfrac{w_x}{u} & 0 & 0
\end{pmatrix}
&
\begin{pmatrix}
0 & 0 & \dfrac{f(v,w)}{u} \\
0 & 0 & 0 \\
-\dfrac{f(v,w)}{u} & 0 & 0
\end{pmatrix}
 \vspace{2ex} \\[1ex] \hline \\[-.5ex]

C^{ij}_{3,6} &
\begin{pmatrix}
\partial_x & 0 & 0 \\
0 & \partial_x & 0 \\
0 & 0 & 0
\end{pmatrix}&
\begin{pmatrix}
0 & 0 & 0 \\
0 & 0 & 0 \\
0 & 0 & 0
\end{pmatrix}
&
\begin{pmatrix}
0 & f(w) & g(w) \\
-f(w) & 0 & c g(w) \\
-g(w) & -c g(w) & 0
\end{pmatrix}
\vspace{2ex} \\[1ex] \hline \\[-.5ex]

C^{ij}_{3,7} &
\begin{pmatrix}
\partial_x & 0 & 0 \\
0 & \partial_x & 0 \\
0 & 0 & 0
\end{pmatrix}
&
\begin{pmatrix}
0 & 0 & 0 \\
0 & 0 & -\dfrac{w_x}{v} \\
0 & \dfrac{w_x}{v} & 0
\end{pmatrix}
&
\begin{pmatrix}
0 & 0 & c f(w) \\
0 & 0 & \dfrac{(1 - c u) f(w)}{v} \\
-c f(w) & -\dfrac{(1 - c u) f(w)}{v} & 0
\end{pmatrix}
\vspace{2ex} \\[1ex] \hline \\[-.5ex] 

C^{ij}_{3,9} & 
\begin{pmatrix}
0 & \partial_x & 0 \\
\partial_x & 0 & 0 \\
0 & 0 & 0
\end{pmatrix} & 
\begin{pmatrix}
0 & 0 & 0 \\
0 & 0 & 0 \\
0 & 0 & 0
\end{pmatrix}
&
\begin{pmatrix}
0 & f(w) & c g(w) \\
-f(w) & 0 & g(w) \\
-c g(w) & -g(w) & 0
\end{pmatrix}
\vspace{2ex} \\[1ex] \hline \\[-.5ex] 

C^{ij}_{3,10} &
\begin{pmatrix}
0 & \partial_x & 0 \\
\partial_x & 0 & 0 \\
0 & 0 & 0
\end{pmatrix}
&
\begin{pmatrix}
0 & 0 & -\dfrac{w_x}{v} \\
0 & 0 & 0 \\
\dfrac{w_x}{v} & 0 & 0
\end{pmatrix}
&
\begin{pmatrix}
0 & f(w) & \dfrac{h(w) - u g(w)}{v} \\
-f(w) & 0 & g(w) \\
-\dfrac{h(w) - u g(w)}{v} & -g(w) & 0
\end{pmatrix}
\vspace{2ex} \\[1ex] \hline \\[-.5ex] 
    \end{array}
\end{equation*}
}
We finally list the operators $C^{ij}_{3,8}$ and $C^{ij}_{3,11}$
{\footnotesize
\begin{equation*}
\begin{split} 
C^{ij}_{3,8} = 
\begin{pmatrix}
\partial_x & 0 & 0 \\
0 & \partial_x & 0 \\
0 & 0 & 0
\end{pmatrix}
&+
\begin{pmatrix}
0 & 0 & -\dfrac{w w_x}{u w - v} \\
0 & 0 & \dfrac{w_x}{u w - v} \\
\dfrac{w w_x}{u w - v} & -\dfrac{w_x}{u w - v} & 0
\end{pmatrix} \\[1ex]
&+ (1 + w^2) f(w)
\begin{pmatrix}
0 & \dfrac{1}{1 + w^2} & \dfrac{w - c v \sqrt{1 + w^2}}{u w - v} \\
-\dfrac{1}{1 + w^2} & 0 & -\dfrac{1 - c u \sqrt{1 + w^2}}{u w - v} \\
-\dfrac{w - c v \sqrt{1 + w^2}}{u w - v} & \dfrac{1 - c u \sqrt{1 + w^2}}{u w - v} & 0
\end{pmatrix}\\[2ex]
C^{ij}_{3,11} = 
\begin{pmatrix}
0 & \partial_x & 0 \\
\partial_x & 0 & 0 \\
0 & 0 & 0
\end{pmatrix}
&+
\begin{pmatrix}
0 & 0 & \dfrac{w_x}{u w - v} \\
0 & 0 & -\dfrac{w w_x}{u w - v} \\
-\dfrac{w_x}{u w - v} & \dfrac{w w_x}{u w - v} & 0
\end{pmatrix} \\[1ex]
&+ f(w)
\begin{pmatrix}
0 & \dfrac{c}{\sqrt{w}} & \dfrac{u w - 2 c \sqrt{w}}{u w - v} \\
-\dfrac{c}{\sqrt{w}} & 0 & -\dfrac{w \left(v - 2 c \sqrt{w}\right)}{u w - v} \\
-\dfrac{u w - 2 c \sqrt{w}}{u w - v} & \dfrac{w \left(v - 2 c \sqrt{w}\right)}{u w - v} & 0
\end{pmatrix}
\end{split} 
\end{equation*}
}  

\vspace{4ex}

\section{Proof of Theorem \ref{2x2thm}}\label{class2x2}
In order to prove the theorem, we first compute the conditions $\Phi^{ijk}=\Phi^{kij}$, that lead to the system
\begin{align}
& \left(a h_{u}^2+b h_{v}^1\right)\omega_{v}+a(  h_{u}^1- h_{v}^2- c_1)\omega_{u}=0, \label{siss1} \\ 
&\left(a h_{u}^2+bh_{v}^1\right)\omega_{u} + b(h_{v}^{2}- h_{u}^{1}-c_1) \omega_{v}=0\,.  \label{siss2}
\end{align}
We consider different cases taking into account the behaviour of the derivatives $\omega_u$ and $\omega_v.$
\begin{enumerate}
    \item { Case $\omega_u = \omega_v = 0$.}\\
  In terms of $h^1$ and $h^2$, starting from \eqref{formadiomega}, this gives
    \begin{equation*}
        \begin{cases} 
            h^1_{uu} + h^2_{uv} = 0 \\[.5ex]
            h^1_{uv} + h^2_{vv} = 0
        \end{cases} \implies~~ \begin{cases} 
            \partial_u(h^1_{u} + h^2_{v}) = 0 \\[.5ex]
            \partial_v(h^1_{u} + h^2_{v}) = 0
        \end{cases} \implies~~ \begin{cases} 
            h^1_{u} + h^2_{v} = f(v) \\[.5ex]
            h^1_{u} + h^2_{v} = g(u)
        \end{cases} 
    \end{equation*}
    with $f(v),g(u)$ arbitrary functions. By comparison, it follows $f(v)=g(u)=\tilde{c}$. \\
    Imposing the conditions prescribed by Corollary $\ref{cormok}$ and with the help of computer algebra, we find that the functions $h_1$ and $h_2$ must be linear, i.e.
    \begin{align*}
        &h^1(u,v)=k_1 u+k_2 v+k_3,\\
        &h^2(u,v)=k_4 u+k_5 v+k_6,
    \end{align*}
    where $k_i$  ($i=1, \dots 6$) are arbitrary constants.
     \item {Case $\omega_u=0, \omega_v\neq 0$} \\   
     In terms of $h^1$ and $h^2,$ we get
\begin{align}\label{1}
    h_{uu}^{1}+h_{uv}^{2}&= 0\,, \\[1ex] 
   c\!\left( h_{uv}^{1}+h_{vv}^{2}\right)&=  \omega_{v}\,, \label{2}
\end{align}
alongside with 
\begin{align}
& a h_{u}^2+b h_{v}^1=0\,,\\[1ex]
& h_{v}^{2}- h_{u}^{1}-c_1=0\, ,
\end{align}
which derived with respect to $u$ and $v$ yield 
\begin{equation}
    \begin{cases}
        a h^2_{uu} + b h^1_{uv} = 0 \\[1ex]
        h_{vu}^{2}- h_{uu}^{1}=0
    \end{cases}\,, \qquad 
    \begin{cases}
        a h^2_{uv} + b h^1_{vv} = 0 \\[1ex]
        h_{vv}^{2}- h_{uv}^{1}=0
    \end{cases}\, .
\end{equation}
Combining the above relations, we get
\begin{equation}
    h^2_{uv} = h^1_{uu} = h^1_{vv} = 0 \,, \qquad h^2_{vv} = h^1_{uv} = -\frac{a}{b}\, h^2_{uu} = \frac{\omega_v}{2c} \,.   
\end{equation}
By integration, we obtain the forms of $h^1$ and $h^2,$ namely 
\begin{align*}
    &h^1= k_1\,u + k_2\,v + k_3\,uv + k_4\,,\\
    &  h^2 = \dfrac{k_3}{2} \,v^2 -\dfrac{b}{2 a}k_3 u^2+ k_5 \,v + k_6\,u + k_7  \, ,  
\end{align*}
with $k_i$, $(i=1, \dots, 7)$ arbitrary constants. 

Computing $R_{12}^{12}$ in \eqref{r1}, we obtain $2 \, b \, k_3^2=0$ and hence $k_3=0$, so that the resulting $h^i$ are once again linear.
  \item {Case $\omega_v=0, \omega_u \neq 0$} \\   This case is perfectly similar to the previous one.
\item {Case $\omega_u \neq 0, \omega_v \neq 0$}\\ 
In this case, with the help of computer algebra, we verified that the system given by \eqref{siss1} and \eqref{siss2} has solutions different from the linear ones only if the arbitrary constant $c_1$ in \eqref{formadiomega} is equal to zero. Under this assumption, \eqref{siss1} and \eqref{siss2} can be rewritten as
 \begin{align}
      &a h_{u}^{2}+b h_{v}^{1}=\dfrac{(h_{v}^{2}-h_{u}^{1}) w_{u}}{w_{v}} \label{siss11}\\
      & (a w_{u}^2+b w_{v}^{2}) (h_{v}^{2}-h_{u}^{1})=0 \label{siss22}
  \end{align}
We recall that it is possible to apply the spectral theorem to the metric $\eta^{ij}$, so that $(a,b)=(1,1),(a,b)=(1,-1), (a,b)=(-1,1)$ and $(a,b)=(-1,-1)$. We can then further distinguish two other subcases: 
\begin{itemize}
    \item[\textit{i})] $ab>0$. We notice that if $a$ and $b$ have the same sign, \eqref{siss11} and \eqref{siss22} are equivalent to the Cauchy-Riemann system for $h^1,h^2$:
\begin{equation} \label{eq:h1h2_constr_case1}
      \begin{cases}
          h_{u}^{2}+ h_{v}^{1}=0\\
       h_{v}^{2}-h_{u}^{1}=0,
      \end{cases}
  \end{equation}
  whose solutions are equivalent to the solutions of the Laplace problem for $h^{1}$ and $h^{2}$:
\begin{equation}
    \Delta h^i=0 \quad \Leftrightarrow \quad h^i_{uu}+h^i_{vv}=0 \qquad i=1,2.
\end{equation} 
Hence we obtain
\begin{equation*}
     h^{1}(u,v)= \xi_1(u+ i v) + \xi_2(u-i v), \quad h^{2}(u,v)= -i \xi_1(u+ i v) + i \xi_2\left(u-i v\right).
\end{equation*}
Finally, computing $R_{12}^{12}$, it immediately follows that either $\xi_1^{''}$ or $\xi_2^{''}$ is equal to zero.
    \item[\textit{ii})] $ab<0$.  If $a$ and $b$ are opposite in sign, the solution to equations \eqref{siss22} is given either by   $h_{v}^{2}-h_{u}^{1}=0$ or $w_u^2-w_v^2=0.$ In the first case, 
    from \eqref{siss11} and \eqref{siss22}, we obtain
  \begin{equation} \label{eq:h1h2_constr_case2}
      \begin{cases}    h_{u}^{2}- h_{v}^{1}=0,\\
       h_{v}^{2}-h_{u}^{1}=0,
       \end{cases}
  \end{equation}
  that is a $p$-system or, equivalently, the wave equation with unitary velocity:
  \begin{equation}
      h^i_{uu}-h^i_{vv}=0, \qquad i=1,2.
  \end{equation}
The solution is explicitly given by
  \begin{equation}
      h^1(u,v)=\xi_1(u+v)+\xi_2(u-v), \qquad h^2(u,v)=\xi_1(u+v)-\xi_2(u-v),
  \end{equation}
  where $\xi_1,\xi_2$ are arbitrary functions in their arguments. Once again, from $R_{12}^{12},$ we get that one of the functions $\xi_1$ and $\xi_2$ must be linear. 
\\

On the other hand, if we assume $a \omega_{u}^{2}+ b \omega_{v}^{2}=0,$ noticing that $b/a= -1$, we can rewrite this condition as
\begin{equation}
    (\omega_{u}- \omega_{v})(\omega_{u}+ \omega_{v})=0.
\end{equation}
This implies $w= F(u \pm v)$, along with 
\begin{equation}
    a h_{u}^{2}+b h_{v}^{1}=\dfrac{(h_{v}^{2}-h_{u}^{1})}{\pm 1}.
    \label{dasostituire1}
\end{equation}
At this point, we consider \eqref{eq2thm}. The only non trivial conditions are
\begin{align}
      &\left(a h_{u}^2+b h_{v}^1\right) (b h_{vv}^2-ah_{uu}^2)+ 2 a b  h_{uv}^2(h_{u}^{1}-h_{v}^{2})=0,  \label{condizionina1}\\
      & \left(a h_{u}^2+b h_{v}^1\right) (b h_{vv}^1-ah_{uu}^1)+ 2 a b  h_{uv}^1(h_{u}^{1}-h_{v}^{2})=0.
      \label{condizionina2}
\end{align}
Substituting \eqref{dasostituire1} in \eqref{condizionina1} and \eqref{condizionina2}, one obtains \begin{equation}
     (h_{v}^{2}-h_{u}^{1})\left(b h_{vv}^2-a h_{uu}^2 \pm 2 h_{uv}^2\right)=0
    \end{equation}
\begin{equation}
    (h_{v}^{2}-h_{u}^{1})\left( b h_{vv}^1-a h_{uu}^1 \pm 2  h_{uv}^1\right)=0
\end{equation}
Hence, remembering that we are in the case $h_{v}^{2}-h_{u}^{1} \neq 0$, we get that $h^1$ and $h^2$ satisfy
\begin{align*}
    & b h_{vv}^1-a h_{uu}^1 \pm 2  h_{uv}^1=0,\\
    & b h_{vv}^2 -a h_{uu}^2 + \pm  h_{uv}^2=0.
\end{align*}
Depending on the sign of $a$ and $b$ and on the choice of $\mathcal{\omega}$, two possible cases arise
\begin{align*}
    &h^1(u,v)=\xi_1(u+v)+(v-u)\,\xi_2(u+v)\\
     &h^2(u,v)=\xi_3(u+v)+(v-u)\,\xi_4(u+v)
\end{align*}
or
\begin{align*}
    &h^1(u,v)=\xi_1(v-u)+(u+v)\,\xi_2(v-u)\\
     &h^2(u,v)=\xi_3(v-u)+(u+v)\,\xi_4(v-u)
\end{align*}
From the tensor $R$, we obtain that $\xi_2=-\xi_4$ for the first choice of $h^1$ and $h^2$, and  $\xi_2=\xi_4$ for the second one.
\end{itemize}
\end{enumerate}

\section{Proof of Theorem \ref{nijnonhom}}\label{nijapp}
To prove the theorem we first compute the Nijenhuis tensor for the affinor $L^i_j=g_{js}(c^{si}_ku^k+f^{si})$. We split the computations into two parts, considering first the coefficients of $u$.  

\vspace*{-2ex}

\begin{enumerate}
    \item {Coefficients of $u$ of degree 1.} 
    
    Collecting the expression of $\mathcal{N}_L$ for $u$ we obtain
\begin{equation}
    g_{js} g_{kl} c_{p}^{li} c_{a}^{sp}- g_{ks} g_{jl} c_{p}^{li} c_{a}^{sp}+ g_{jl} g_{ps} c_{a}^{si} c_{k}^{lp}- g_{kl} g_{ps} c_{a}^{si} c_{j}^{lp}=0,
\end{equation}
that is equivalent to 
\begin{equation}
     g_{jl} c_{a}^{lp} g_{ks} c_{p}^{si} - g_{ks}  c_{a}^{sp} g_{jl} c_{p}^{li}+ g_{ps} c_{a}^{si}[g_{jl}  c_{k}^{lp}- g_{kl} c_{j}^{lp}]=0.
\end{equation}
Using the compatibility condition between the metric and the Poisson tensor, i.e.  $g^{is}c^{jk}_s+g^{js}c^{ik}_s=0$, applying $g_{ai}g_{bj}$ to both the sides, we have $ g_{al} c_{b}^{lp} + g_{bl} c_{a}^{lp}=0$, so that  from the previous expression we get
\begin{equation}
     g_{jl}  g_{ks} [c_{a}^{lp} c_{p}^{si} - c_{a}^{sp}  c_{p}^{li}]+ 2 g_{ps} c_{a}^{si}  g_{jl} c_{k}^{lp}=0.
\end{equation}
From the Jacobi identity on the structure constants and using the skew-symmetry property on the constants we can substitute $c_{a}^{lp} c_{p}^{si} - c_{a}^{sp}  c_{p}^{li}=c_{a}^{lp} c_{p}^{si} + c_{a}^{sp}  c_{p}^{il}= - c_{a}^{ip}  c_{p}^{ls}$, obtaining
\begin{equation}
    -g_{jl} g_{ks} c_{a}^{ip} c_{p}^{ls}+ 2  g_{ps} c_{a}^{si}  g_{jl} c_{k}^{lp}=0.
\end{equation}
or equivalently
\begin{equation}
    -g_{jl} g_{ks} c_{a}^{ip} c_{p}^{ls}-   g_{ps} c_{a}^{si}  g_{jl} c_{k}^{pl} - g_{ps} c_{a}^{si} g_{kl} c_{j}^{lp} =0.
    \label{ul}
\end{equation}
Now, using some simple algebra, we get the additional relations
\begin{equation*}
 g_{ps} c_{a}^{si}  g_{jl} c_{k}^{pl}=-g_{kp} c_{s}^{pl}  g_{jl} c_{a}^{si},   \qquad g_{ps} c_{a}^{si} g_{kl} c_{j}^{lp}=g_{ps} c_{a}^{si} g_{jl}c_{k}^{pl},
\end{equation*}
so that \eqref{ul} can be rewritten as
\begin{equation}
      -g_{jl} g_{ks} c_{a}^{ip} c_{p}^{ls} + g_{jl}g_{ks} c_{p}^{sl}   c_{a}^{pi}+g_{k s} c_{p}^{sl} c_{a}^{pi} g_{jl}=0.
\end{equation}
In the first term, $c_{a}^{ip}c_{p}^{ls} =- c_{a}^{sp} c_{p}^{il}-c_{a}^{lp} c_{p}^{si} =+c_{a}^{sp} c_{p}^{li}-c_{a}^{lp} c_{p}^{si}$, so that
\begin{equation}
    g_{jl} g_{ks} c_{a}^{lp} c_{p}^{si}- g_{jl} g_{ks} c_{a}^{sp} c_{p}^{li} + g_{jl}g_{ks} c_{p}^{sl}   c_{a}^{pi}+g_{k s} c_{p}^{sl} c_{a}^{pi} g_{jl}=0.
\end{equation}
With similar manipulations, we recover
\begin{equation}
   g_{ap} g_{ks} c_{l}^{ps} c_{j}^{li}+ g_{jl}g_{ks} c_{p}^{sl}   c_{a}^{pi}+g_{jl} g_{ks} c_{a}^{sp} c_{p}^{li}=0
\end{equation}
Finally, in the last term, $g_{jl} c_{a}^{sp} c_{p}^{li}=-g_{pl} c_{a}^{sp} c_{j}^{li}=g_{lp} c_{a}^{ps} c_{j}^{li}=- g_{ap} c_{l}^{ps} c_{j}^{li}$ and we get
\begin{align}\begin{split}
   0&= g_{ap} g_{ks} c_{l}^{ps} c_{j}^{li}+ g_{jl}g_{ks} c_{p}^{sl}   c_{a}^{pi}- g_{ap} g_{ks} c_{l}^{ps} c_{j}^{li}= g_{jl}g_{ks} c_{p}^{sl}   c_{a}^{pi}.\end{split}
\end{align}
We then conclude that the terms in $u$ annihilate if and only if
condition \eqref{cond1x} holds true.

\item Coefficients of $u$ of degree 0. 

We now focus on the coefficients of $f$. We have
    \begin{equation}
       f^{ls} [ g_{il} g_{jb} c_{s}^{bk}-g_{jl} g_{ib} c_{s}^{bk}]+ g_{ib} g_{js} c_{l}^{sb} f^{lk}-g_{sl} g_{jb} c_{i}^{bs} f^{lk} \label{eqf1}
    \end{equation}
   Since $f$ is a cocycle, in the third term
    we can use $
        c_{l}^{sb} f^{lk}=-c_{l}^{bk} f^{ls}-c_{l}^{ks} f^{lb}
   $. Hence, in \eqref{eqf1}
       \begin{equation}
       f^{ls} [ g_{il} g_{jb} c_{s}^{bk}-g_{jl} g_{ib} c_{s}^{bk}]- g_{ib} g_{js}c_{l}^{bk} f^{ls}- g_{ib} g_{js}c_{l}^{ks} f^{lb} -g_{sl} g_{jb} c_{i}^{bs} f^{lk}=0. \label{eqf2}
    \end{equation}
  Simplifying the second and the third term (that turn out to be opposite with some algebra), 
  we get with some simple algebra
        \begin{equation}
      g_{il} g_{jb} c_{s}^{bk} f^{ls}- g_{ib} g_{js}c_{l}^{ks} f^{lb}+g_{is} g_{jb}c_{l}^{bk} f^{ls}+g_{is} g_{jb}c_{l}^{ks} f^{lb}=0.
      \label{eqf5}
         \end{equation}  
     Manipulating the third term $g_{is} g_{jb}c_{l}^{bk} f^{ls}=g_{il} g_{jb}c_{s}^{bk} f^{sl}=-g_{il} g_{jb}c_{s}^{bk} f^{ls}$ and simplifying, it remains
         \begin{equation}
    - g_{ib} g_{js}c_{l}^{ks} f^{lb}+g_{is} g_{jb}c_{l}^{ks} f^{lb}=0
      \label{eqf6}
         \end{equation}  
  This implies the second condition and proves the Theorem.
  \end{enumerate}      

\end{appendices}

\end{document}